\newif\ifNOTtalk
\NOTtalktrue
\documentclass[letterpaper, 11pt]{article}
\usepackage{latexsym}
\usepackage{amssymb}
\usepackage{times}
\usepackage[in]{fullpage}
\usepackage{amsmath,amsfonts,amsthm}
\usepackage{graphicx}
\usepackage{algpseudocode}
\usepackage{algorithm}
\usepackage{comment}
\newcommand{\ble}{\begin{lemma}}
\newcommand{\ele}{\end{lemma}}

\newtheorem{lemma}{Lemma}[section]
\newtheorem{problem}[lemma]{Problem}

\newtheorem{theorem}[lemma]{Theorem}

\newtheorem{definition}[lemma]{Definition}
\newtheorem{corollary}[lemma]{Corollary}

\newtheorem{fact}[lemma]{Fact}
\newtheorem{note}[lemma]{Note}
\newtheorem{observation}[lemma]{Observation}

\newtheorem{algorithm1}[lemma]{Algorithm}

\newcommand{\beao}{\begin{eqnarray*}}
\newcommand{\eeao}{\end{eqnarray*}\noindent}
\newcommand{\beam}{\begin{eqnarray}}
\newcommand{\eeam}{\end{eqnarray}\noindent}

\newcommand{\one}{{\bf 1}}

\begin{document}

\title{Approximating Large Frequency Moments with $O(n^{1-2/k})$ Bits}
\author{Vladimir Braverman\thanks{Johns Hopkins University, Department of Computer Science. Email: vova@cs.jhu.edu. This work was supported in part by DARPA grant N660001-1-2-4014. Its contents are solely the responsibility of the authors and do not represent the official view of DARPA or the Department of Defense.}, Jonathan Katzman\thanks{Johns Hopkins University, jkatzma2@jhu.edu}, Charles Seidell\thanks{Johns Hopkins University, cseidel5@jhu.edu. This work was supported in part by Pistritto Fellowship.}, Gregory Vorsanger\thanks{Johns Hopkins University, Department of Computer Science. gregvorsanger@jhu.edu}}

\maketitle

\begin{abstract}
In this paper we consider the problem of approximating frequency moments in the streaming model.
Given a stream $D = \{p_1,p_2,\dots,p_m\}$ of numbers from $\{1,\dots, n\}$, a frequency of $i$ is
defined as $f_i = |\{j: p_j = i\}|$. The $k$-th \emph{frequency moment} of $D$ is defined as $F_k = \sum_{i=1}^n f_i^k$.

In this paper we give an upper bound on the space required to find a $k$-th frequency moment of $O(n^{1-2/k})$ bits that matches, up to a constant factor, the lower bound of   \cite{Woodruff:2012:TBD:2213977.2214063} for constant $\epsilon$ and constant $k$.
Our algorithm makes a single pass over the stream and works for any constant $k > 3$.
\end{abstract}

\section{Introduction}\label{introsec}

The analysis of massive datasets has become an exciting topic of theoretical algorithms research. As these datasets grow increasingly large, the need to develop new algorithms which can run using sublinear memory has become paramount.
It is often convenient to view such datasets as \emph{data streams}.
In this paper we consider the following streaming model:
\begin{definition}
Let m and n be positive integers. A \textbf{stream} $D = D(n,m)$ is a sequence of integers $p_1,\ldots,p_m$, where $p_i \in \{1,\ldots,n\}$. A \textbf{frequency vector} is a vector of dimensionality n with non-negative entries $f_i, i\in [n]$ defined as:
$$
f_i = |\{j: 1 \le j \le m, p_j = i\}|
$$
%
A k-th \textbf{frequency moment} of a stream D is defined by $F_k(D) = \sum_{i\in[n]} f_i^k$. Also, $F_\infty = \max_{i\in[n]} f_i$ and $F_0 = |\{i: f_i>0\}|$.
\end{definition}
In their celebrated paper, Alon, Matias, and Szegedy \cite{ams} introduced the following problem:
\begin{problem}
 What is the space complexity of computing a $(1\pm \epsilon)$-approximation of $F_k$ in one pass over $D$?
\end{problem}
In this paper we consider the case where $k>2$. Many algorithms have been designed to solve this particular problem, and we now provide a brief overview of the upper and lower bounds provided. To begin, \cite{ams,ams1} gave a lower bound of $\Omega(n^{1-5/k})$ (for $k\ge 6$) and an upper bound of $O({1\over \epsilon^2}n^{1-1/k}\log(nm))$.
Bar-Yossef, Jayram, Kumar, and Sivakumar \cite{frequency_lower_bound1} improved the lower bound and showed a bound of $\Omega(n^{1-(2+\lambda)/k})$ for their one pass algorithm where $\lambda$ is a small constant. They also showed a lower bound of $\Omega(n^{1-3/k})$ for a constant number of passes.
Chakrabarti, Khot, and Sun \cite{frequency_lower_bound2} showed a lower bound of $\Omega(n^{1-2/k})$ for one pass and $\Omega(n^{1-2/k}/(\log n))$ for a constant number of passes. Gronemeier \cite{DBLP:journals/corr/abs-0902-1609} and Jayram \cite{Jayram1} extended the bound of \cite{frequency_lower_bound2} from one pass to multiple passes.
Woodruff  and Zhang \cite{Woodruff:2012:TBD:2213977.2214063} gave a lower bound of $\Omega(n^{1-2/k}/ (\epsilon^{(4/p)}t))$ for a $t$-pass algorithm.
Ganguly \cite{DBLP:journals/corr/abs-1201-0253} improved the result of  \cite{Woodruff:2012:TBD:2213977.2214063} for small values of $\epsilon$ and for $t=1$. Price and Woodruff \cite{DBLP:conf/isit/PriceW12} gave a lower bound on the number of linear measurements.

In terms of upper bounds, Ganguly \cite{frequency_impr1} and Coppersmith and Kumar \cite{frequency_impr2} simultaneously gave algorithms with space complexity\footnote{The standard notation $\tilde{O}$ hides factors that are polylogarithmic in terms of $n,m$ and polynomial in terms of the error parameter $\epsilon$.} $\tilde{O}(n^{1-1/(k-1)})$.
In their breakthrough paper, Indyk and Woodruff \cite{frequency} gave the first upper bound that is optimal up to a polylogarithmic factor. Their bound was improved by a polylogarithmic factor by Bhuvanagiri, Ganguly, Kesh, and Saha \cite{frequency1}.
Monemizadeh and Woodruff \cite{Monemizadeh:2010:RLP:1873601.1873693} gave a bound of ${O}(\epsilon^{-2}k^2n^{1-2/k}\log^5(n))$ for a $\log(n)$-pass algorithm.
For constant $\epsilon$, Braverman and Ostrovsky \cite{recursive} gave a bound of
${O}(n^{1-2/k}\log^2(n)\log^{(c)}(n))$ where $\log^{(c)}(n)$ is the iterated logarithm function.
Andoni, Krauthgamer, and Onak \cite{andoni_robert_onak} gave a bound of $O(k^{2}\epsilon^{-2-6/p}n^{1-2/k}\log^2(n))$.
Ganguly \cite{DBLP:journals/corr/abs-1104-4552} gave a bound of $O(k^2\epsilon^{-2}n^{1-2/k}E(k,n) \log (n) \log (nmM)/\min(\log (n),\epsilon^{4/k-2}))$ where $E(k,n) = (1-2/k)^{-1}(1-n^{-4(1-2/k)})$.
Braverman and Ostrovsky \cite{recursive1,DBLP:journals/corr/abs-1212-0202} gave a bound of
${O}(n^{1-2/k}\log(n)\log^{(c)}(n))$.

\subsection{Main Result} \label{Main Result}
For constant $\epsilon$ and $k$ we provide a streaming algorithm with space complexity $O(n^{1-2/k})$. Thus, our upper bound matches the lower bound of Woodruff and Zhang \cite{Woodruff:2012:TBD:2213977.2214063} up to a constant factor. Our algorithm makes a single pass over the stream and works for constant $k > 3$.

The main technical contribution is a new algorithm that finds heavy elements in a stream of numbers. Then, combining this result with the broader Martingale Sketches technique from Section \ref{dsjkfkjsdfkjdsfkjsdfkj} we create an algorithm to approximate $F_k$.
In particular, we show:
\begin{theorem}\label{ekmrmflkerfflkreflker}
Let $\epsilon$ be a constant and $k \ge 7$. There exists an algorithm that outputs a $(1\pm \epsilon)$-approximation of $F_k$, makes three passes over the stream, uses $O(n^{1-2/k})$ memory bits, and errs with probability at most $1/3$.
\end{theorem}

\noindent We now present the necessary definitions and theorems.
\begin{definition}\label{fdljdfljkdflkjgdflk}
Let $D$ be a stream and $\rho$ be a parameter. The index $i\in [n]$ is a \emph{$\rho$-heavy element} if
$
f_i^k \ge \rho F_k.
$
\end{definition}

\begin{definition}\label{dsjfnkjdfkjsdfkjf}
A randomized streaming algorithm $\mathcal{A}$ is an \emph{Algorithm for Heavy Elements (AHE)} with parameters
$\rho \text{ and } \delta$ if the following is true:
$\mathcal{A}$ makes three passes over stream $D$ and outputs
a sequence of indices and their frequencies
such that if element $i$ is a $\rho$-heavy element for $F_k$ then $i$ will be one of the indices returned\footnote{
Indices of non-heavy elements can be reported as well.}.
$\mathcal{A}$ errs with probability at most $\delta$.
\end{definition}

\begin{theorem}\label{wdfkjnwkejfkjwefkjwef}
Let $k \ge 7$.
There exists an absolute constant $C\le 10$ and an AHE algorithm with parameters $\rho$ and $\delta$ that uses
\begin{equation}\label{welfjlweflkwflwke}
O({1\over \rho^C}(F_0(D))^{1-2/k}\log{{1\over \delta}})
\end{equation} bits.
\end{theorem}

\begin{theorem}\label{sdsdkflksdflksdf}
Given Theorem \ref{wdfkjnwkejfkjwefkjwef}, for any $\epsilon$ there exists an algorithm that uses
\begin{equation}\label{lkdflgdlfjgjkldfgkjdfgkjdfgkjd}
O({1\over \epsilon^{2C}}(F_0(D))^{1-2/k})
\end{equation}
memory bits, makes three passes over $D$, and outputs a $(1\pm \epsilon)$-approximation of $F_k$ with probability at least $2/3$.
Here $C$ is the constant from Theorem \ref{wdfkjnwkejfkjwefkjwef}.
\end{theorem}

\noindent From here, we see that the main theorem, Theorem \ref{ekmrmflkerfflkreflker}, follows directly from Theorem \ref{sdsdkflksdflksdf}.

After establishing the matching bound with three passes, we improve our algorithm further:
\begin{theorem}\label{ekmrmflkerfflkreflker1}
Let $\epsilon$ be a constant and $k > 3$. Assuming that $m$ and $n$ are polynomially far, there exists an algorithm that outputs a $(1\pm \epsilon)$-approximation of $F_k$, makes one pass over the stream, uses $O(n^{1-2/k})$ memory bits, and errs with probability at most $1/3$.
\end{theorem}

\subsubsection*{Additional results}
The previous theorems demonstrate the optimal reduction from the problem of computing frequency moments for constant $k>2$ to the problem of finding heavy elements with constant error. The Martingale Sketches technique is an improvement over the previous method of recursive sketches \cite{recursive1}. Thus, our method is applicable in a general setting of approximating $L_1$-norms of vectors which have entries obtained by applying entry-wise functions on the frequency vector. As a result, we answer the main open question from \cite{recursive1} and improve several applications in \cite{recursive1}. We will provide a detailed list of these results in the full version of the paper.

\subsection{Roadmap}
In Section \ref{dsfsdlkjfjkljsdflkjsdf}, we prove Theorem \ref{wdfkjnwkejfkjwefkjwef}.
Initially, in order to construct the proof, we make several assumptions which are shown in Table \ref{table1}. Later, in Section \ref{sdflklkjsdfjldsjf}, we show how these assumptions can be removed.
In Section \ref{dsjkfkjsdfkjdsfkjsdfkj} we present the new technique used for proving Theorem \ref{sdsdkflksdflksdf}. This new method, Martingale Sketches, allows the reduction of the problem of computing frequency moments to the problem of finding heavy hitters.
In Section \ref{sdjfkjsdfkjskdjfkjsdf} we prove Theorem \ref{sdsdkflksdflksdf} by combining the Martingale Sketches algorithm with the results from Section 2 and by verifying the space requirements of the final algorithm.
Finally, in Section \ref{sdljfksjdfjksdfkjsd} we prove Theorem \ref{ekmrmflkerfflkreflker1}.
In the remainder of this section we will discuss the related work and will provide an intuition for the main steps of our algorithms and their analysis.

\subsection{Related work}
Approximating $F_k$ has become one of the most
inspiring problems in streaming algorithms. To begin, we provide an incomplete
list of papers on frequency moments
 \cite{stable,5215,ams,ams1,
frequency_lower_bound1, frequency_lower_bound2, 711822, frequency,
frequency_impr2, 776778, 796530, frequency_impr1, 1459774, 1496816, recursive,
nelson, 1807094, 1807101, 982817, 1250891, 1374470, 1265565,DBLP:journals/corr/abs-1104-4552, DBLP:journals/corr/abs-1201-0253, Woodruff:2012:TBD:2213977.2214063, Jayram:2011:OBJ:2133036.2133037} and
references therein.
These and other papers have produced many beautiful
results, important applications, and new methods. Below we will mention a few of the results that provide relevant bounds.
We refer a reader to \cite{strbook, DBLP:reference/db/Woodruff09} and references therein for further details.

In \cite{ams}, the authors observed that it is possible
to approximate $F_2$ in optimal polylogarithmic space.
Kane, Nelson and Woodruff \cite{1807094} gave a space-optimal solution for $F_0$.
Kane, Nelson, and Woodruff \cite{nelson}
gave optimal-space results for $F_k, 0<k<2$.
In addition to the original model of \cite{ams}, a variety of different models of streams have been introduced.
These models include the turnstile model (that allows insertion and deletion) \cite{stable}, the sliding window model \cite{Braverman:2007:SHS:1333875.1334202}, and the distributed model \cite{gibbons,Woodruff:2012:TBD:2213977.2214063,Cormode:2011:CDM:2031792.2031793}.
In the turnstile model, where the updates can be integers in the range $[-M,M]$, the latest bound by Ganguly \cite{DBLP:journals/corr/abs-1104-4552} is $$O(k^2\epsilon^{-2}n^{1-2/k}E(k,n) \log (n) \log (nmM)/\min(\log (n),\epsilon^{4/k-2}))$$ where $E(k,n) = (1-2/k)^{-1}(1-n^{-4(1-2/k)})$. This bound is roughly $O(n^{1-2/k}\log^2(n))$ for constant $\epsilon, k$.
Recently, Li and Woodruff provided a matching lower bound for $\epsilon < 1/(\log n)^{O(1)}$ \cite{Rag2013}. Thus, for the turnstile model, the problem has been solved optimally for $\epsilon < 1/(\log n)^{O(1)}$ \cite{DBLP:journals/corr/abs-1104-4552,Rag2013}.  These results combined with our result demonstrate that the turnstile model is fundamentally different from the model of Alon, Matias, and Szegedy.


\subsection{An Illustrative Example}
In this section we will demonstrate the main steps of our method by considering a simplified problem.
Let $D$ be a stream with the following promise: all non-zero frequencies are equal to $1$ with the exception of a single element
$i$ such that the frequency of $i$ is $f_i\ge n^{1/k}$.
Furthermore, $m=\Theta(n)$ and if we split $D$ into intervals of length $O(n^{1-1/k})$ then $i$ appears once in each interval.
Clearly, $i$ is the heavy element and the goal of the algorithm will be to find the value of $i$.
This simplified case is interesting because the same promise problem is used for the lower bound in \cite{frequency_lower_bound2} and in many other papers. We will thus illustrate the capability of our method by showing that a bound  $O(n^{1-2/k})$ is achievable in this case.

We will assume without loss of generality that $i=1$. This assumption does not change the analysis but simplifies our notation. In \cite{DBLP:journals/corr/abs-1212-0202} it is shown that  $O(n^{1-2/k})$ samples are sufficient to solve the problem.
However, each sample requires $\log n$ bits for identification (we will use a notion of ``ID'' to identify the value of $i\in [n]$.)
As well, any known algorithm stores information about the frequency of the heavy element.
This can be done by storing a sketch or an explicit approximate counter. In the most direct implementation, $\log m$ bits are required to store the counter. In this example we will assume that $\log n = \Theta (\log m)$ and we will use a single parameter $\log n$.

If $n^{1-2/k}$ independent samples are sampled from each interval then the probability to sample $1$ is a constant.
Next, observe that most of the time only $O(1)$ bits are needed for the counters since all frequencies except $i=1$ are either zero or one. Thus, it is sufficient to reduce the bits for IDs.

The key idea is to replace IDs with signatures and uniform sampling with (appropriately chosen) hashing.
Combining signatures of constant length with hashing ensures that the number of false positives is relatively small.
Specifically, consider a hash function $g:[n]\mapsto [n^{1-1/k}]$ and let the $z$-th sample of the $i$-th interval be defined as follows. Let
\begin{equation}\label{sdflksdfljjdfsd}
\Gamma_{i,z} = \{j: g(p_j)=z\}
\end{equation}
where $p_j$ are elements from the $i$-th interval.
To obtain the final sample, we sample one element uniformly at random from $\Gamma_{i,z}$.
We call this sampling schema \emph{two-level sampling}.
It is not hard to see that the probability that $1$ is sampled using the new sampling method is still a constant.
Now consider the case that each sample is represented using a signature of length $O(1)$.
Suppose that we store signature $SIG$ for the $z$-th sample in the $i$-th interval.
The comparison of the sample with another element $q$ of the stream will be defined by the following procedure.
We say that they are equal if $g(q) = z$ and the signature of $q$ is equal to $SIG$.
Consider the case when we sample the heavy element. In this case the consecutive appearances of $1$ will always be
declared equal to the sample.
Consider the case when $l$ has been sampled and when $f_l=1$. The probability that there will be any collision
in the next interval is at most $2^{-|SIG|}$.
Therefore we can exploit the probability gap between these two cases.

Specifically, deleting samples with a small number of collisions allows for increasing signatures for the remaining samples in the future intervals.
After $2^\gamma$ intervals, it is possible to increase the signature by $O(1)$ bits for $\gamma = 1,2,\dots$.
Simple analysis shows that the heavy element will never be discarded and that the number of active samples decreases exponentially with $\gamma$.
Thus, the total expected space for storing the data is
$O\left({n^{1-2/k}\gamma\over 2^{O(\gamma)}}\right).$
The aforementioned procedure is called the $\gamma$-th round for the $i$-th interval.
At any moment there are at most $2^\gamma$ intervals in the $\gamma$-th round and the total space is $O(n^{1-2/k})$.
For $\gamma = \Omega(\log\log n)$ storing IDs instead of signatures implies that
if the heavy element is not discarded then the correct answer is produced.
The algorithm works in one pass and uses $O(n^{1-2/k})$ bits\footnote{It is possible to show that $g$ can be pairwise independent.}.



\subsection{Intuition}\label{dfsdfsefsdfsdf}

\subsubsection{High Level Description of the Algorithm}

We present a composite algorithm to estimate frequency moments. At the absolute lowest degree of detail, we perform three steps. First, we determine the length of the stream. Second, we use a new algorithm to efficiently find heavy elements. Finally, we use a new technique to estimate the value of frequency moments from the weight of the found heavy elements.
We now describe the intuition of each of these parts in detail.

\subsubsection{The Heavy Hitter Algorithm}
The key step in our algorithm for frequency moment computation is a new technique to compute the heavy hitters of a stream.
In order to determine which elements are $\rho$-heavy in stream $D$, we present an algorithm that is implemented as a sequence of sub-algorithms, and in general we will refer to each of these sub-algorithms as a ``game".
In \cite{DBLP:journals/corr/abs-1212-0202} it is shown that  $O(n^{1-2/k})$  samples are sufficient to solve the problem.
However, each sample requires $\log n$ bits for counting the frequency and for identifying the elements. The resulting bound is $O(n^{1-2/k}\log n)$ bits. The goal of our algorithm, therefore, is to reduce the space required for counters and IDs from $\log n$ to an amortized $O(1)$ bits, achieving the optimal bound.

First we will describe the workings of a single instance of the game, and then we will describe the sequence of games that composes our heavy hitter algorithm.  Each game in the sequence will be run in parallel, and the cost of the sequence of games will form a geometric series, which when evaluated will yield a total cost of $O(n^{1-2/k})$ bits. The crucial observation is that a heavy element in the stream will be returned by at least one of these games with constant probability, and will be sufficiently frequent to stand out from the other returned values as the true heavy hitter.

\subsubsection*{The Game}

To find a heavy element of a stream and prove Theorem \ref{wdfkjnwkejfkjwefkjwef}, we play a game using the stream as input. First we split the stream into equally sized rows as we read it in, and assemble them into a matrix $M$.

\includegraphics[scale=0.75]{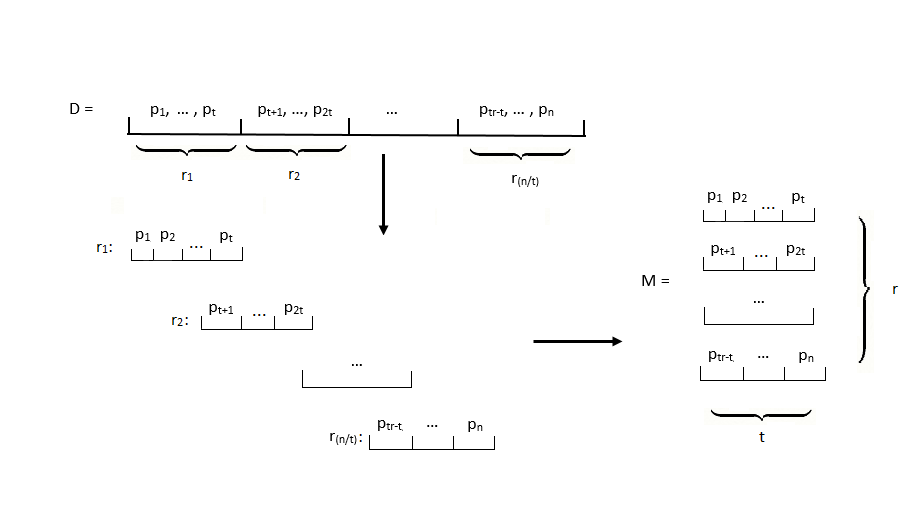}
$$\text{Picture 1: Transforming the stream into a matrix}$$

 A single game is described colloquially as follows:  for each row, we create a ``team" that is composed of a group of $w$ players each competing to be the winner of that game. To create these teams, we sample elements from the current row to act as the players on each team, and give each player an ID number equal to one of the sampled elements.
For each player on a team, maintain a counter to track how often their ID number appears as we move through the stream.  If the player's counter does not grow fast enough, that player is removed from the game.

The $\gamma$-th round is played by each team after $2^\gamma$ rows have passed since the team started playing. In each round, we divide the players of each team into groups of size $3^\gamma$, the players compete within these groups, and there is at most one winner per group, i.e. the surviving player whose counter is highest. The winning player from each group continues to play throughout the remainder of the game, competing in further rounds. Players who are not winners withdraw from the game and do not compete in any further rounds.

At the end of the game, each team will have at most one winner. The winners from every team then compete against each other, and the player with the highest overall counter is the overall winner. We define this game more formally in section 2.1.

\begin{note}\label{sdfdfsdffsdf}
We only need to consider $\gamma \le \log_3 n^{1-2/k}$ since for any larger values of $\gamma$ there is only a single winner of the game in the $i$-th row. Since we keep the winner until the end of the game, there is no need to check the expiration of row winners.
\end{note}

Assume that $F_k = O(n)$ and that $1$ is a heavy element that appears among every $O(n^{1-1/k})$ elements. We can make two observations.
First, the counter of the player who samples $1$ requires only $O(\gamma)$ bits after seeing $2^{\gamma}n^{1-1/k}$ elements of the stream.
Also, this counter will have a nice property of linear growth: after seeing $2^{\gamma}$ intervals the counter will be at least
$2^{\gamma}$.

Second, we can observe that the sum of the frequencies of every element that has frequency larger than $\lambda$ is at most $G_k\over \lambda^{k-1}$, where $G_k = F_k - f_1^k$.
This is because
$$
\sum_{l\ge 2, f_l\ge \lambda} f_l \le {1\over \lambda^{k-1}}\sum_{l\ge 2, f_l\ge \lambda} f_l^k \le {G_k \over \lambda^{k-1}}.
$$
Thus, the number of intervals with many such elements is small. For example, let an element $l$ be ``$\gamma$-bad'' if $f_l\ge 2^\gamma$ and consider an interval to be a ``bad" interval if it contains more than $n^{1-1/k}\over 2^{100\gamma}$ distinct bad elements. There are at most ${G_k\over 2^{(k-1)\gamma}}{2^{100\gamma}\over n^{1-1/k}}$ such intervals.
Under the assumption that $F_k = O(n)$, and for sufficiently large $k$, this number is exponentially smaller than $n^{1/k} \over 3^{\gamma}$.
As a result, the probability that a bad element will be sampled among the $3^\gamma$ competitors is exponentially small.

Unfortunately, the above observations are not true in general. First, the distribution of the heavy element throughout the stream can be arbitrary. For example, half of the appearances of the heavy element may occur in a single row and
thus we need $\log n$ bits at the time each player starts playing the game. Second, it is possible that $G_k$ is much larger than $n$ in which case the number of
bad intervals can be larger.
It is possible that there exist intervals with the number of $1$s being $2^i$ for every $i=0,1,\dots, \lfloor\log(n)\rfloor$ and they comprise an equal percentage of the total frequency.

To overcome these problems we show that there exists a $\beta$ such that there are a sufficiently large number of intervals where the number of $1$s in each interval is in the range $[2^\beta, 2^{\beta+1}]$.

In general, our goal is to show that for any distribution of the heavy element in the stream
there exists some $\beta$ such that
\begin{enumerate}
\item $O\left(n^{1-2/k} \over 2^{\mu\beta}\right)$ samples are needed to sample the heavy element with a constant probability, where $\mu$ is a small constant, and,
\item players that may compete with the heavy element will expire with high probability before the competition.
\end{enumerate}

The space bound implies that the problem can be solved without knowledge of the value of $\beta$.

\subsubsection*{A Sequence of Games}
An exhaustive search of the range of $\beta$, which we will later show is at most logarithmic, yields the sequence of games that eventually constitutes our heavy hitter algorithm. As stated before, we will prove that the cost for these games is geometric and, after some slight modifications discussed in this paper, yields the desired overall cost.

Our proof of correctness will rely heavily on what we term the ``Noisy Lemma" (\ref{dfdfljknjlfkjlrfekjf}).  In this lemma we aim to show that at least one event in a collection of ``good" events will come to pass with at least a certain probability, even if each event is impeded by a number of ``noisy" events that can prevent the event from occurring.  This lemma can then be applied to show that at least one player corresponding to the actual heavy hitter will win overall, even if there is a chance that other players with large but not heavy elements will win some games.

Having established this algorithm, it is clear that we have solved the problem of the space required by counters. It remains to show that we can store the ID of each player in sufficiently small space to achieve our desired bound. In order to do this, we will transfer the duty of tracking the identity of each player from a deterministic ID to a hashed signature.

\subsubsection*{Signatures instead of IDs}\label{adfsdfsefsdfsdf}
Given a new element of the stream, our algorithm needs to be able to differentiate elements for the following reasons:
\begin{itemize}
\item If the new element has the same ID as one of the samples, then the stored counter of the sample should be incremented.
\item If the new element has been chosen as a new sample for one of the players, it is necessary to compare
the IDs of the new elements and the current sample. If they are the same, we increment the counter; if they are different, we have to replace the sample.
\end{itemize}

Since there are $n$ possible elements, $\log n$ bits are required to identify all of the IDs deterministically.
However, note that after $O(\log\log n)$ rounds a team with initially $w$ active players will only have $w\over \log^{\Omega(1)} n$ active players.
Thus, $O(n^{1-2/k})$ bits are sufficient to store all IDs of sampled elements in all tables for all old rows for which at least $O(\log\log n)$ rounds have passed.

Therefore, we only need to take care of the first $\log\log n$ rounds each row plays.
It is acceptable to err with a small probability.
Thus, we can use random signatures to represent IDs to use less space. Our goal is to reach $O(\gamma^{O(1)})$ bits per signature.
Unfortunately, if we simply hash $[n]$ into a range of $[2^{\gamma^{O(1)}}]$, the number of collisions per row will still be polynomial in $n$
for small $\gamma$'s, which does not help.

In general, a small (constant) probability of collision can be shown only for sets of small cardinalities
when hashing is applied on a set of $2^{\gamma^{O(1)}}$ elements.
Thus, to use signatures we have to reduce the cardinality of the set of competitors. We do so
by implementing a sampling procedure with an additional independent hash function.
We choose the function carefully so that in a game with parameter $\beta$ the probability to sample the heavy element for any row from $S_{\beta}(\alpha)$ is preserved (see Table \ref{table2}).
First, we hash elements into a range $g: [n] \mapsto [t_{\alpha}]$, where $t_{\alpha}$ is the number of columns in our matrix, and allow only those elements with values smaller than $2^{\beta}$
to be sampled. We then compute signatures only for the elements in the ``pool'' $\Gamma$ of all elements that pass the $g$ filter.
With constant probability, $|\Gamma| = O(2^{\beta})$ and no element will have the same signature as the heavy element.
This sampling procedure will be implemented in Section \ref{dskfjkjasdskjsdf}.

The same argument will work for any $2^{\gamma}$ rows if the length of the signature is $\Omega(\gamma)$.
Thus, we can use $\gamma$ bits to represent all of the IDs. Then, after $\log\log n$ rounds, the cardinality will be small enough such that we are able to switch our method and use the real ID of a given player's element. We implement this by assigning the ID of the first element that can be sampled and has a matching signature to the player, and then counting based on this new ID. With constant probability this will be the same heavy element that we used to generate the signature to begin with.

While this technique reduces the space required, the downside is that there will be collisions for some of the $w$ players and as a result we need to overcome two technical issues.
First, due to multiple IDs being hashed to the same signature, the counters of the players can be larger than the frequency of the sampled element they are supposed to be counting. Second, if the heavy element is sampled from row $i$ it can now be incorrectly compared with many non-heavy elements from rows $\{i, \dots, i+2^{\gamma}\}$ that collide with another value initially sampled in row $i$. Intuitively, this can cause the counter for a given signature to be large due to many non-heavy elements hashing to the same signature. Because much of the analysis on the correctness of the algorithm is based on the counters of players who have sampled non-heavy elements, this difference must be addressed as well.

We overcome both of these problems as follows.

First, after we have progressed far enough to assign the real ID in addition to the signature, we will add a new counter. We will stop incrementing the old counter, and the new counter will count only the frequency of  elements with the chosen ID. Thus, we will no longer be counting based on the signature, and we will separate the values counted by the signature from the values counted by the ID.
Then, after more rounds,  we will switch to using only the new counter and thus the first problem will be fixed. This change creates an additional problem: some appearances of the heavy element might be discarded. We will ensure that the new counter will be polynomially larger than the old counter at the time when it will be discarded. Thus, the change is negligible and will not affect the correctness.

Second, we will prove that the probability that the counter of a non-heavy competitor increases by enough to impact the game is exponentially small in terms of $\gamma$. Thus, the issue of competitor collision is solved and the same analysis still applies.

Therefore, by adding the use of hashed signatures to the way we differentiate elements, we will show that we can bound the amount of bits used to store all ID numbers and all signatures by $O(n^{1-2/k})$.

\subsubsection{Martingale Sketches}

Now that we have an algorithm that can detect heavy hitters in $O(n^{1-2/k})$ space, it remains to show that this directly yields a method for approximating the $k$-th frequency moment.  We refer to the process by which this occurs as Martingale Sketches.  These sketches are constructed using a martingale sequence of random variables.  Our new method of approximation rests upon another result of this section: a reduction up to a constant factor of the problem of $k$-th frequency moments to the problem of heavy hitters.

Consider a vector where the sum of its elements cannot be computed directly.  If the elements of the vector vary in magnitude, then some elements will have a larger impact on the sum than others.  Now consider a second vector made from including or excluding each element of the first by the repeated flip of a fair coin, and then doubling the value of every included element.  The expected difference between the sums of the two vectors is $0$.  But because of the disproportionate contribution of heavy hitters, the actual difference will most likely not be $0$.  If we can find the heavy hitters of the first vector, we can examine which ones were included and which were excluded in the second vector.  Intuitively, the excluded ones will increase the difference between the vector sums, while the included heavy hitters will decrease it (because of the scaling up by a factor of $2$, and their already large contribution to the total sum).  This allows us to approximate the difference between the two vector sums.  If we repeat this process for the second vector and a new vector made from including or excluding each of its elements (with the included elements having their values doubled), and so on, then the repeated differences along with the sum of the final vector can be used together to accurately approximate the sum of the first vector.  Thus, finding a frequency moment is reducible to finding the heavy hitters of a series of vectors.

While the overall idea of reducing a vector sum to its heavy hitters is not new, what our algorithm provides is a cost function that is geometric by the nature of the given reduction.  Thus, the total space cost required for these computations matches the lower bound for frequency computation, up to a constant factor.

In the section on Martingale Sketches, we will show that one can view this process as the construction of a martingale sequence of random variables dependent only on finding the heavy hitters of given vectors.

\section{Finding Heavy Elements and Proving Theorem \ref{wdfkjnwkejfkjwefkjwef}}\label{dsfsdlkjfjkljsdflkjsdf}

Without loss of generality, suppose that the number $1$ is a heavy element in stream $D$. We will omit floors and ceilings unless they are necessary.

\subsection{Initial Algorithm}\label{sdkjnfkasddksak}

We will begin our solution by designing an algorithm for finding heavy elements in a stream which conforms to several assumptions.
Later, in Section \ref{sdflklkjsdfjldsjf} we will show how to remove these assumptions, which are listed in Table \ref{table1}.


The key step of the algorithm is a subroutine that we call a \emph{game}, which is described in Section \ref{F}.
The algorithm will execute (in parallel) a sequence of several  games with different parameters $\alpha$ and $\beta$.
The high level description is given in Algorithm \ref{sdjnfsdfsjdfskjd}.

\bigskip
{\begin{algorithm}\caption{Sequence of Games}\label{sdjnfsdfsjdfskjd}

\begin{enumerate}

\item For integer $\eta = 0,\dots, RANGE$.

\begin{enumerate}
\item For integer $u=\lceil1.5\eta\rceil,\dots,20\eta$
\begin{enumerate}
\item Put $\alpha = -0.5\eta$, $\beta =u$
\item Play the Game with parameters $\alpha, \beta$.
\end{enumerate}

\item For integer $u=20\eta+1\dots RANGE$
\begin{enumerate}
\item Put $\alpha = u/5$,
\item For integer $\beta = 0.8u-0.5\eta-2,\dots, u$

    \begin{enumerate}
    \item Play the Game with parameters $\alpha, \beta$.
    \end{enumerate}

\end{enumerate}

\end{enumerate}

\end{enumerate}
\end{algorithm}
}

The winner of each $(\alpha,\beta)$-game will compete against the others, and the overall winner will be the output of the algorithm. This brings us to one of our main technical results of this section which is given in Theorem \ref{sdfsdfsdfsdfsaldkflkaslkflsda}.
Informally, we show that the heavy hitter will be the winner with high probability.

\subsubsection{The $(\alpha,\beta)$-Game}\label{F}
Let $\chi \in \mathbb{Z}$. Define matrix $M_{\chi}$ by splitting $D$ into
$r_\chi$ consecutive intervals\footnote{Without loss of generality we will assume that $F_1$ is divisible by $t_{\chi}$. If this is not the case, we always can ignore the last incomplete row. If there is at least $0.5f_1$ appearances of $1$ in the incomplete row then it is possible to find the heavy element using $o(n^{1-2/k})$ bits. Otherwise, the problem is reduced to the problem when $m$ is divisible by $2^\chi n^{1/k}$. Also, w.l.o.g., we assume that $t_{\chi}$ is an integer. Otherwise, replacing $t_{\chi}$ with $\lceil t_{\chi}\rceil$ will work.} and mapping the $i$-th interval into the $i$-th row\footnote{See Tables \ref{table1}, \ref{table2}, \ref{table3} for the definitions.}. Let $t_{\chi}$ be the number of columns in $M_{\chi}$.

Given two parameters, $\alpha$ and $\beta$, we play the ($\alpha, \beta$)-game as follows.
Let $M_{\alpha}$ be a matrix with $r_\alpha$ rows and $t_\alpha$ columns.
The algorithm reads in the stream and preforms a simple transformation to represent it as rows of the matrix. When the algorithm reads row $i$, it selects a \emph{team} of \emph{players} from this row to be team $i$.
Each player on the the $i$-th team represents a sampled element from the $i$-th row which is chosen by non-uniform \emph{two-level sampling} that is described in detail in Section \ref{dskfjkjasdskjsdf}. Each player on team $i$ maintains a counter that is incremented whenever another element is found from row $j \ge i$ that shares its ID.

After the team is chosen, it plays \emph{rounds} of the game. During each round, a player may become \emph{inactive}. At the end of the game at most one player will be active.
Specifically, the $\gamma$-th round is played by each team after $2^\gamma$ rows have passed since the team started playing.

An element that becomes inactive is said to expire. A player can become inactive in two ways:

First, there are special restrictions on each player's counter. After each round, each player must have a counter that is greater than the threshold $TR$, which is a function on $\gamma, \alpha,$ and $\eta$, in order to continue playing. If the counter is smaller than $TR$ then the player becomes inactive. In addition, each player requires its initial counter (after reading its own row) to be at least a parameter $IC$.

Second, the players compete with each other. The competition in the $\gamma$-th round is defined as follows.
The players are divided into groups of size $3^\gamma$. The active player from each group with the highest counter is declared the winner of its group.
All other players become inactive (the ties are broken arbitrarily).
Also, the $\beta$-th round is played immediately after reading the first row and no rounds $\gamma < \beta$ are played. 


%
%

\subsubsection{Parameters and Notations.}

\begin{definition}\label{rewfkwejrkjwefkjwekjwer}
We say an $(\alpha,\beta)$-game is a successful game if the winner of the game is $1$ and the counter of the winner is at least $0.5f_1$. We say that a sequence of $(\alpha,\beta)$-games is a successful sequence if the winner of the sequence is $1$ and the counter of the winner is at least $0.5f_1$.
\end{definition}

\begin{fact}\label{asljffkjsadfkjds}
Let $HALF$ be the set of the first half of all appearances of $1$ in $D$ (the first $\frac{f_1}{2}$ appearances).
Consider a sequence of $(\alpha,\beta)$-games. If any element from $HALF$ wins in
its team in any game of the sequence then the sequence is successful.
\end{fact}
\begin{proof}
Indeed, at the end of the game the counter of $1$ will be at least $0.5f_1$ by the definition of $HALF$. Recall that $0.5f_1 > f_l$ for any $l\neq 1$.
Note that the game never overestimates counters and thus the claim follows.
Therefore the fact is correct.
\end{proof}
Thus, we will restrict our analysis to $HALF$. In Table \ref{table2} we establish common notations and constants.

\begin{table}[H]

\centering
\begin{tabular}{|l|l|}
\hline
$G_k$ & $F_k - f_1^k$, The frequency moment of all non-heavy elements \\ \hline
$t_{\chi}$ & $2^{-\chi}n^{1-1/k}$, the number of columns in $M_{\chi}$ \\ \hline
$r_{\chi}$ & $F_1/t_{\chi}$, the number of rows in $M_{\chi}$ \\ \hline
$f_{l}(\chi,i)$ & the number of times $l$ appears in the $i$-th row of $M_{\chi}$ \\ \hline
$S_u(\chi)$ & the set of rows $i$ in $M_{\chi}$ with $2^{u-1}\le f_{1}(\chi,i) < 2^{u}$\\ \hline
$T_\lambda$ & ${  \{  l :  f_l > \lambda, l > 1 \}}$, the set of elements $l\neq 1$ whose frequency are larger than $\lambda$. \\ \hline
$\mu = 2^{-10}$  & A constant used for analysis \\ \hline
$\Psi = \mu^{-6}(\log\left({2^\mu\over 2^\mu-1}\right) + 100+k)$  & A constant used for analysis \\ \hline
\end{tabular}
\caption{Some Notations in Section \ref{dsfsdlkjfjkljsdflkjsdf} }\label{table2}
\end{table}

In Table \ref{table3} we describe the values of various parameters that are related to an $(\alpha,\beta)$-game.
The parameters are functions of $\alpha,\gamma,\beta,u,\eta$.

\newcommand{\xixi}
{
$3^{-\gamma}2^{-\beta-\mu\gamma}$
}
\newcommand{\TR}
{
$2^{\gamma - \alpha+\eta-1}$
}

\begin{table}[H]
\centering
\begin{tabular}{|l|l|}
\hline
$w$ & $2^{-\mu\beta}n^{1-2/k}$ \\ \hline
$CM$ & $2^{\gamma^2}$ \\ \hline
$IC$ & $2^{\beta-7}$ \\ \hline
$TR$ (Threshold)& \TR \\ \hline
$RANGE$ (Range for $\alpha,\beta$-games)& $O(\log\log n)$\\ \hline
$\xi$    (A parameter that is used in Definition \ref{dsfsefsdfsdfdsfdsfdsfdsf}.) &\xixi \\ \hline\end{tabular}
\caption{Parameters of a single game.}\label{table3}
\end{table}

\subsubsection{Assumptions}\label{sdljfkjsdfkjsdflkjsfd}

In Table \ref{table1} we describe the initial assumptions (and related parameters).

\begin{table}[H]

\centering
\begin{tabular}{|l|l|}
\hline
$F_1 \le C_2n$  & $C_2$ is some absolute constant. \\ \hline
$0\le \eta \le \log\log n$ & Assume that $\eta$ is even.\\ \hline
$G_k \le 2^{\eta k}n$ &  \\ \hline
$f_1 \ge C_22^{\eta+2\Psi+1}n^{1/k}$ & See Table \ref{table3}.\\ \hline
If $1$ is sampled then &\\
$1$ will not expire. &\\
That is, for any $\gamma$-th round  &\\
 the value of the counter will be at least $TR$ &\\ \hline
\end{tabular}
\caption{Assumptions}\label{table1}
\end{table}

The assumptions in this section will be removed in Section \ref{sdflklkjsdfjldsjf}.

\subsubsection{The Two-Level Sampling Method for Players}\label{dskfjkjasdskjsdf}

\includegraphics[scale=0.75]{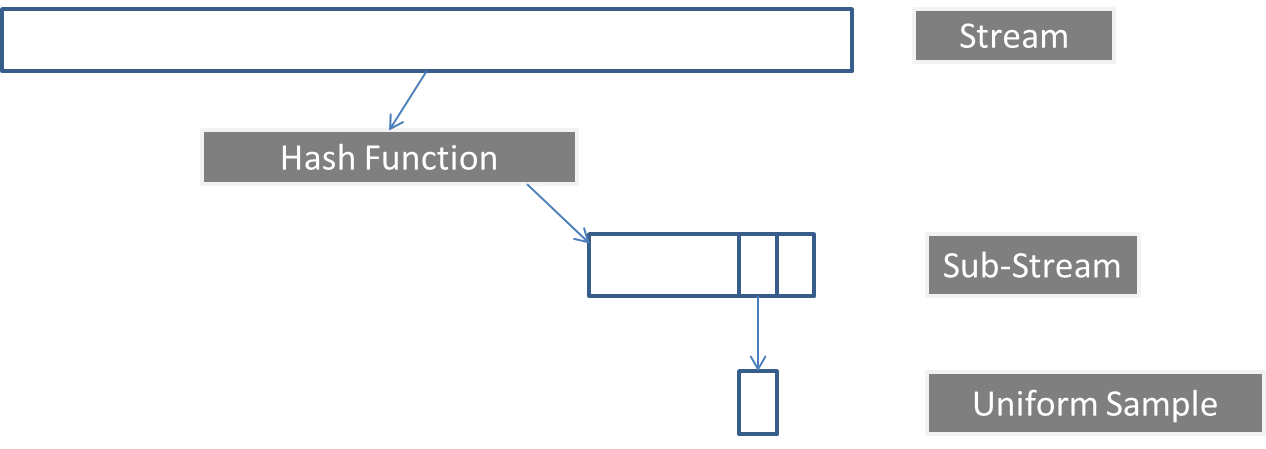}
$$\text{Picture 2: Two-Level Sampling}$$

To save space we need to define a two-level sampling method (see Picture 2) that combines hashing and uniform random sampling from the stream.
Informally, we hash the stream and then sample uniformly from the resulting substream. Formally, the algorithm is as follows:

\bigskip
{\begin{algorithm}[H] \caption{Two-Level Sampling$(Q, \lambda)$}\label{flkdgldfgjgldfjg}

\begin{enumerate}

\item Generate pairwise independent hash function $h: [n]\mapsto \{0,1\}$ such that $P(h(i)=1) = \lambda$.

\item For every element of $p\in Q$: if $h(q)=1$ then add $q$ to a ``pool'' $\Gamma$.

\item In parallel, maintain a uniform random sample $L$ from $\Gamma$ using reservoir sampling \cite{reservoir}.

\item Return $L$.

\end{enumerate}
\end{algorithm}
}

The technical claim that we will need in the analysis of our game is the following.

\begin{lemma}\label{wedfnkjkjkwefkjwef}
Let $Q$ be a stream of length $K$. Let $L$ be the two-level sample from Algorithm \ref{flkdgldfgjgldfjg}.
Let $B_l$ be the event that $l\in \Gamma$ and let $A$ be the event that $L \neq 1$.
If $f_1 \ge \lambda T K$ then
$$
P(A \mid B_1) \le 2T^{-1/2}.
$$
\end{lemma}
\begin{proof}
Let $Z = \sum_{l=2}^n f_lX_l$ where $X_l$ is the indicator of the event that $h(l)=1$. Let $U$ be the event that $Z \le \lambda T^{1/2} K$.
It is sufficient to prove two inequalities:
\begin{equation}\label{edfkjkdskvckjdsv}
P(B_1 \cap \bar{U}) \le T^{-1/2}\lambda.
\end{equation}

and

\begin{equation}\label{fvdfgsdfsdfsfsfd}
P(A \cap B_1 \cap U) \le T^{-1/2}\lambda.
\end{equation}
Indeed, if both inequalities are true then
\begin{equation}\label{dsfklmslkdflsdflsd}
P(A \cap B_1) \le P(A \cap B_1 \cap U) + P(B_1 \cap \bar{U}) \le 2T^{-1/2}\lambda,
\end{equation}
and
$$
P(A|B_1) = {P(A\cap B_1)\over P(B_1)} \le {2T^{-1/2}\lambda\over \lambda} = 2T^{-1/2},
$$
which gives us the statement of the lemma.

To show $(\ref{edfkjkdskvckjdsv})$,
define $W = ZX_1$. By pairwise independence $E(W) \le \lambda^2K$. By the Markov inequality $P(W \ge \lambda T^{1/2} K) \le T^{-1/2}\lambda$.
It is straightforward to verify that events $(B_1 \cap \bar{U})$ and $(W > \lambda T^{1/2} K)$ are equivalent. Thus $(\ref{edfkjkdskvckjdsv})$ follows.

To show $(\ref{fvdfgsdfsdfsfsfd})$, observe that
\begin{equation}\label{fdvdksjvksdkdfjdkjfkdfj}
P(A \cap B_1 \cap U) = P(A \cap B_1 \cap (Z \le \lambda T^{1/2} K)) = \sum_{z=0}^{\lfloor \lambda T^{1/2} K \rfloor} P(A \cap B_1 \cap (Z = z)).
\end{equation}
Here the first equality follows from the definition of $U$, the second equality follows from the fact that $Z$ is a random variable with positive integer values.
Therefore the definition of conditional probability, together with $(\ref{fdvdksjvksdkdfjdkjfkdfj})$, implies:
\begin{equation}\label{fdvdksjvksdkdfjdkjfkdfj1}
P(A \cap B_1 \cap U) = \sum_{z=0}^{\lfloor \lambda T^{1/2} K \rfloor} P(A \mid B_1 \cap (Z = z))P(B_1 \cap (Z = z)).
\end{equation}

Fix $z$ such that $0\le z\le \lambda T^{1/2}K$ and consider the following bound on $P(A \mid B_1 \cap (Z = z))$.
If $X_1 = 1$ and $Z = z$ then $L$ is a sample from a bag of $z+f_1$ elements where $f_1$ of elements are equal to $1$.
Therefore:
\begin{equation}\label{fefgfergergerfgerg}
P(A \mid B_1 \cap (Z = z)) = {z\over z+f_1}.
\end{equation}
Further, our choice of $z$ and the bound on $f_1$ that is given in the statement of the lemma imply:
${z\over z+f_1} \le \lambda {T^{1/2}K\over \lambda T K}.$ Combining the last inequality with $(\ref{fefgfergergerfgerg})$ gives
\begin{equation}\label{lfdslflsdlsdlkjd}
P(A \mid B_1 \cap (Z = z)) \le T^{-1/2}.
\end{equation}
If we substitute $(\ref{lfdslflsdlsdlkjd})$ into $(\ref{fdvdksjvksdkdfjdkjfkdfj1})$ we obtain:
$$
P(A \cap B_1 \cap U) \le T^{-1/2}\sum_{z=0}^{\lfloor \lambda T^{1/2} K\rfloor} P(B_1 \cap (Z = z)) \le T^{-1/2}P(B_1) = T^{-1/2}\lambda.
$$
Thus, $(\ref{fvdfgsdfsdfsfsfd})$ is correct.
\end{proof}

In an ($\alpha,\beta$)-game we will use the two-level sampling method to determine the ID number for a single player.
The sampling will depend on an additional parameter $\beta$.
This sampling method will be instrumental in reducing the space complexity. The intuition is explained in Section \ref{dfsdfsefsdfsdf}.

\begin{definition}\label{sdfkljsdfkjksdkjf}
Let $i$ be a row index and $z$ be a player index.
Define $L$ to be a sample from the $i$-th row as follows. Let $g$ be a uniform pairwise independent hash function $g: [n] \mapsto [t_{\alpha}]$.
First\footnote{Note that the sample is defined for all players since $w2^{\beta} = o(t_{\alpha})$ for our parameters.}, define the $i,z$-pool to be:
\begin{equation}\label{dsfsdfsdfsdfsdfsdf}
\Gamma_{i,z} = \{j: 2^{\beta-\Psi}(z-1) \le g(m_{i,j}) < 2^{\beta-\Psi}z\}.
\end{equation}
Second, define $L = m_{i,J}$ where $J$ is a uniform random sample from $\Gamma_{i,z}$ and $m_{i,J}$ is the J-th entry of the i-th column of matrix $M_\alpha$. Note that hash functions for distinct rows are independent.
\end{definition}

\begin{lemma}\label{fkljsdfkjlsdflkjdsf}
Let $i,z$ be fixed. Denote  events
$$A= \{L\neq 1\}, B_l= \{\exists j \in \Gamma_{i,z} \text{ with } m_{i,j}=l\}.$$
If $f_1(\alpha,i) \ge 2^\beta$ then
$$
P(A | B_1) \le 0.02.
$$

\end{lemma}
\ifNOTtalk
\begin{proof}
We will apply Lemma \ref{wedfnkjkjkwefkjwef} with $\lambda = {2^{\beta-\Psi}\over t_\alpha}$, $T=10^4$ and $K=t_\alpha$.
Recall that the definition of $\Psi$ in Table \ref{table2} implies that $2^\Psi > 10^4$.
Thus, $f_1(\alpha,i) \ge 2^\beta\ge \lambda T K$. Thus, by Lemma \ref{wedfnkjkjkwefkjwef} $P(A | B_1) \le 2T^{-1/2} = 0.02.$

\end{proof}

\subsection{Correctness of the Algorithm}\label{dsfknskdjfkjsdfjksdfjk}

We prove the correctness of our algorithm by using the following technical lemma.

\subsubsection{The Noisy Lemma}\label{adsfkjsdkjafkjsdfkj}

In the following lemma $A_1,\dots, A_N$ represent ``good'' events. We aim to bound the probability that at least one good event occurs.
Each good event occurs if and only if  a ``basic'' event $C_i$ occurs and several ``noisy'' events $\{{B}_{i,j}\}_{j=1}^L$ do not occur.

\begin{lemma}\label{dfdfljknjlfkjlrfekjf}
Let $A_1,\dots, A_N$ be a set of events where
$$
A_i = C_i \bigcap (\cap_{j=1}^{L} \bar{B}_{i,j})
$$
and where $C_i$ and $B_{i,j}$ are events.
Suppose\footnote{Two natural cases are when the events are independent or when they are disjoint.} that for all $i\neq j\in [N]$
\begin{equation}\label{efjwfjkjdfkfdkj}
P(C_i \cap C_j) \le P(C_i)P(C_j).
\end{equation}
If
\begin{equation}\label{kdfmgvlkdfgkdfglkdf}
\sum_{i=1}^N P(C_i) = a,
\end{equation}
and for any $i\in [N]$
\begin{equation}\label{dfmdflgdflkg}
\sum_{j=1}^L P(B_{i,j}|C_i) \le b
\end{equation}
then
\begin{equation}\label{dfkvldkflkdf}
P(\cup_{i=1}^N A_{i}) \ge a(1-0.5a-b).
\end{equation}

\end{lemma}
\ifNOTtalk
\begin{proof}
The Bonferroni inequality (see, e.g., $(2.1)$ in \cite{naiman1992inclusion}) implies:
\begin{equation}\label{sddjsflsdfjklsdfjsd}
P(\cup_{i=1}^N A_{i}) \ge \sum_{i=1}^N P(A_{i}) - \sum_{1\le i< l \le N} P(A_{i} \cap A_l).
\end{equation}

Let us bound the second sum of the right-hand side in $(\ref{sddjsflsdfjklsdfjsd})$.
Since for all $i\in [N]$ we have $A_i \subseteq C_i$ and by $(\ref{efjwfjkjdfkfdkj})$:
\begin{equation}\label{dsflksdlkfsdlfksldk}
P(A_{i} \cap A_l) \le P(C_{i} \cap C_l) \le P(C_{i})P(C_l).
\end{equation}
Thus,
\begin{equation}\label{fsdlvdlfsdlkfj1}
\sum_{1\le i< l \le N} P(A_{i} \cap A_l) \le \sum_{1\le i< l \le N} P(C_{i})P(C_l) \le 0.5\left(\sum_{i=1}^N P(C_{i})\right)^2 = 0.5a^2.
\end{equation}
Here the first inequality follows from the Bonferroni inequality, the second inequality follows from direct computations and the equality follows from $(\ref{kdfmgvlkdfgkdfglkdf})$.
In summary, we obtain the bound on the second sum of the right-hand side in $(\ref{sddjsflsdfjklsdfjsd})$:
\begin{equation}\label{fsdlvdlfsdlkfj}
\sum_{1\le i< l \le N} P(A_{i} \cap A_l) \le 0.5a^2.
\end{equation}

Let us bound the first sum of the right-hand side in $(\ref{sddjsflsdfjklsdfjsd})$. Fix $i\in [N]$ and observe that
\begin{equation}\label{sdfksdflksdflksdf}
P(\cup_{j=1}^{L} B_{i,j} | C_i) \le \sum_{j=1}^L P(B_{i,j} | C_i) \le b.
\end{equation}
Here the first inequality follows from union bound
and the second follows from $(\ref{dfmdflgdflkg})$.
Thus
\begin{equation}\label{kjsdcvksdskddskf}
P(\cap_{j=1}^{L} \bar{B}_{i,j}| C_i) = 1-P(\cup_{j=1}^{L} {B}_{i,j} | C_i) \ge 1-b.
\end{equation}
Therefore
\begin{equation}\label{dfkgmdlfkgfdlkgdfkl}
 P(A_{i}) = P(C_i \bigcap (\cap_{j=1}^{L} \bar{B}_{i,j})) = P(\cap_{j=1}^{L} \bar{B}_{i,j}| C_i)P(C_i) \ge (1-b)P(C_i),
\end{equation}
where the last inequality follows from $(\ref{kjsdcvksdskddskf})$.
Summing over all $i\in [N]$ we obtain:
\begin{equation}\label{dfkgmdlfkgfdlkgdfkl1}
\sum_{i=1}^N P(A_{i}) \ge (1-b)\sum_{i=1}^N P(C_{i}) = (1-b)a,
\end{equation}
where the equality follows from $(\ref{kdfmgvlkdfgkdfglkdf})$. In summary, we obtain the bound on the first sum of the right-hand side in $(\ref{sddjsflsdfjklsdfjsd})$:
\begin{equation}\label{dfkgmdlfkgfdlkgdfkl2}
\sum_{i=1}^N P(A_{i}) \ge (1-b)a,
\end{equation}

Combining $(\ref{sddjsflsdfjklsdfjsd})$,$(\ref{dfkgmdlfkgfdlkgdfkl2})$ and $(\ref{fsdlvdlfsdlkfj})$ we obtain a bound that is equivalent to $(\ref{dfkvldkflkdf})$:
\begin{equation}\label{sdlsdfjsdfkjsdfkljsdfkljdsfkjlsdf}
P(\cup_{i=1}^N A_{i}) \ge (1-b)a - 0.5a^2.
\end{equation}
\end{proof}
\fi

\begin{corollary}\label{egtrweryrthrthhrtdhrtht}
Consider the definitions and assumptions of Lemma \ref{dfdfljknjlfkjlrfekjf} and assume also that $b<0.1$. Then
\begin{equation}\label{fdlkvlkjdfvlkdfgvlkjdf}
P(\cup_{i=1}^N A_{i}) \ge \min \{ 0.8a, 0.04\}.
\end{equation}
\end{corollary}
\ifNOTtalk
\begin{proof}
If $a\le 0.1$ then $(\ref{dfkvldkflkdf})$ implies
\begin{equation}\label{equation1023}
P(\cup_{i=1}^N A_{i}) > 0.8a.
\end{equation}
Suppose that\footnote{In general it is possible that $a \ge 1$.} $a\ge 0.1$. If there exists $i$ such that $P(C_i)\ge 0.05$ then
by using $(\ref{dfkgmdlfkgfdlkgdfkl})$ we obtain
$$
P(A_i) \ge (1-b)P(C_i) \ge 0.9*0.05 = 0.045.
$$
Lastly, suppose that for all $i$ it is true that $P(C_i) < 0.05$. In this case there exists a subset $I \subseteq [N]$ such that
$$
0.05 \le \sum_{i\in I}P(C_i) \le 0.1.
$$
Denote $a' = \sum_{i\in I}P(C_i)$ and apply Lemma \ref{dfdfljknjlfkjlrfekjf} to the set of events defined by $I$.
By $(\ref{equation1023})$ we have that
$$
P(\cup_{i=1}^N A_{i}) \ge P(\cup_{i\in I} A_{i}) \ge 0.8a' \ge 0.04.
$$
\end{proof}
\fi

\subsubsection{Conditions for Winning the Game}\label{dslkjkasdjksjdfkjsdg}

In this section we will state and prove the sufficient conditions for winning the game.
The main idea is to identify a sequence of events for which Corollary \ref{egtrweryrthrthhrtdhrtht} from Section \ref{adsfkjsdkjafkjsdfkj} is applicable. Here, the two-level sampling from Section \ref{dskfjkjasdskjsdf} will be instrumental.
Specifically, we will establish that if the frequency is sufficiently large and if the parameters of the game, $\alpha$ and $\beta$, are chosen carefully then
the probability to win the game is bounded from below (up to a constant factor) by the probability to sample the heavy element into
one of the hashing pools $\Gamma_{i,z}$ from Definition \ref{sdfkljsdfkjksdkjf}. As a result, the probability of success becomes a constant for the right parameters. In Section \ref{sdfgsdgsdgsdg} of the appendix we show the existence of a pair $\alpha, \beta$ that will satisfy the conditions of this section. As a result by the exhaustive search in Algorithm \ref{sdjnfsdfsjdfskjd} it follows that $1$ will be sampled and will become a winner in its team for at least one pair $\alpha,\beta$ and, by Fact \ref{asljffkjsadfkjds}, will become a winner of all games.

Recall that the definitions of the variables used below are given in Tables \ref{table1} and \ref{table2}. In this section we assume that $\alpha$ and $\beta$ are fixed and $\gamma$ is a parameter.

\begin{definition}\label{ewfwefwefwef}
Let $i\in [r_\alpha]$ be the $i$-th row of $M_\alpha$. We say that $i$ is a $(\lambda, \phi, \tau)$-dense row if elements from $T_\lambda$ with high row frequency compose at least a $\phi$ fraction of all elements from this row:
\begin{equation}\label{sdfljdsfljsdff}
| \{  l: f_{l}(\alpha,i) > \tau, l \in T_{\lambda} \}| > t_{\alpha} \phi.
\end{equation}
\end{definition}

\begin{definition}\label{dsfsefsdfsdfdsfdsfdsfdsf}
Let $i\in [r_\alpha]$ be an index of a row in $M_\alpha$. We say that $i$ is $\beta$-bad if there exists $\gamma \ge \beta$ such that $i$ is $(TR, \xi, \beta-7)$-dense\footnote{Recall that $TR =$\TR and $\xi=$\xixi. See Table \ref{table3}.}.
Also, $i$ is a $\beta$-great if it is not $\beta$-bad.
When the values of $\alpha$ and $\beta$ are clear from the context we simply say ``bad'' or ``great'' row.
\end{definition}


\begin{lemma}\label{sdfljsdldjfljeflkjfd}
Consider an $(\alpha,\beta)$-game. If there exists $X$ great rows that are from $S_{\beta}(\alpha)$ and an absolute constant $BC$ such that
\begin{equation}\label{absBC}
\frac{wX2^\beta }{t_{\alpha}} \ge BC
\end{equation}
then there exists another absolute constant $GC$ such that the probability that the heavy element will beat all teammates in at least one row
is at least $GC$.
\end{lemma}
\ifNOTtalk
\begin{proof}
We will prove the lemma by applying Corollary \ref{egtrweryrthrthhrtdhrtht}.
Specifically, let $i\in S_{\beta}(\alpha)$ also be an index of a great row and let $z$ be a player's index such that $z\in [w]$.
Denote by
$
C_{i,z}
$
the event that $1$ will be sampled by the hash function into the $z$-th pool\footnote{We refer the reader to explanation of the two-level sampling method in Section \ref{dskfjkjasdskjsdf}.} $\Gamma_{i,z}$.
Specifically, following $(\ref{dsfsdfsdfsdfsdfsdf})$ in Definition \ref{sdfkljsdfkjksdkjf}:
$$
C_{i,z} = \{(z-1)2^{\beta-\Psi} \le g(1) < z2^{\beta-\Psi}\}.
$$
Denote by
$
H_{i,z}
$
the event that $1$ is not sampled from the pool $\Gamma_{i,z}$ to become the sample\footnote{Recall that this means that the uniform random sample from $\Gamma_{i,z}$ is not $1$.} of the $z$-th player in the $i$-th row.
Fix $\gamma\ge \beta$. By the description of the $(\alpha,\beta)$-game in Section \ref{F}, there exist at most $3^\gamma$ players that can compete with the $i$-th player in the $\gamma$-th round.
Denote by
$
B_{i,z,\gamma}
$
the event that at least one such player samples one of the elements from $T_{TR}$.
Let $U_{i,z}$ be the event that
the initial counter of the sample has value smaller than $IC = 2^{\beta-7}$.
Let
$$
A_{i,z} = C_{i,z} \bigcap \bar{H}_{i,z} \bigcap \bar{U}_{i,z}  \bigcap ( \cap_{\gamma\ge \beta}\bar{B}_{i,z,\gamma}).
$$

Observe that if at least one of the $A_{i,z}$ is true then a player that samples $1$ will beat all other players on the same team.
This statement follows from the description of the game in Section \ref{F}. Indeed, by Assumption \ref{wdkjwefkjwe} the counter will never go below the threshold $TR$ and the player will never expire. Also, event $\bar{U}$ implies that the player will not be discarded right away. Event $\bar{H}_{i,z}$ implies that $1$ will be sampled from the pool and will become the sample of the $z$-th player.
Event $\bar{B}_{i,z,\gamma}$ implies that no one will compete with the $z$-th player in the $\gamma$-th round.
Thus, $A_{i,z}$ implies that $1$ will be the winner in its team.
Thus, our goal is to obtain a constant lower bound on the probability that at least one ``good'' event happens, $P(\cup_{i,z} A_{i,z})$.

Let us now show that the other events satisfy the premises of Corollary \ref{egtrweryrthrthhrtdhrtht}.
First, observe that $(\ref{efjwfjkjdfkfdkj})$ is correct. Indeed, if $(i,z)\neq (i',z')$ then
$$
P(C_{i,z},C_{i',z'}) \le P(C_{i,z})P(C_{i',z'}).
$$
This is because if $i\neq i'$ then the events are independent, and if $i=i', z\neq z'$ then the events are mutually exclusive.

Second, let us bound the conditional probabilities.
For all events $U$ we have:
$$
P(U_{i,z} \mid C_{i,z}) \le 0.02.
$$
This is because all appearances of $1$ have the same chances to be sampled. There are at least $2^{\beta-1}$ appearances
of $1$ since $i\in S_{\beta}(\alpha)$. Thus, the probability that the counter is smaller than $2^{\beta-7}$ is at most ${1\over 2^6} \le 0.02$.

For $B$'s we have the following:
Since $i$ is a great row, the number of elements that will not become inactive at the $\gamma$-th round is at most
$t_{\alpha}\over 3^{\gamma}2^{\beta+\mu\gamma}$. For any fixed element, the probability to be sampled is at most $2^{\beta-\Psi} \over t_{\alpha} $
by Definition \ref{sdfkljsdfkjksdkjf}.
Therefore by union bound, the probability that any of the $3^{\gamma}$ players will survive to compete with  the player that samples the heavy element, $1$, is at most
$$
3^\gamma  {t_{\alpha} \over 3^{\gamma}2^{\beta+\mu\gamma}}  {2^{\beta-\Psi} \over t_{\alpha} } \le {1\over 2^{\Psi}}{1 \over 2^{\mu\gamma}}.
$$
Then, summing up over all $\gamma \ge \beta$ we obtain:
$$
\sum_{\gamma\ge \beta} P(B_{i,z,\gamma}| C_{i,z}) \le {1\over 2^{\Psi}}{2^{\mu}\over {2^{\mu}-1}} \le 0.01.
$$

Here, the first inequality follows from direct computation and the second inequality follows from substitution of $\mu$ and $\Psi$.
For $H$'s we have the following:
by Lemma \ref{fkljsdfkjlsdflkjdsf} we have that
$$
P(H_{i,z}\mid C_{i,z}) \le 0.02.
$$
Thus
$$
P(U_{i,z} \mid C_{i,z}) +P(H_{i,z}\mid C_{i,z}) + \sum_{\gamma\ge \beta} P(B_{i,z,\gamma}| C_{i,z}) < 0.1,
$$
and therefore the conditions of Corollary \ref{egtrweryrthrthhrtdhrtht} are satisfied.

It remains to show that the sum of the probabilities of the events $C_{i,z}$ is bounded from below by a constant.
Let $I$ be the set of all great rows in $S_{\beta}(\alpha)$. Note that by definition $|I| = X$, and also $P(C_i,z) = {2^{\beta-\Psi}\over t_{\alpha}}$ by definition of our sampling method from \ref{dskfjkjasdskjsdf}. Indeed,

\begin{equation}
\sum_{i\in I,z \in [w]} P(C_{i,z}) = wX{2^{\beta-\Psi}\over t_{\alpha}} \ge {BC\over 2^{\Psi}}.
\end{equation}

The last inequality follows from equation \ref{absBC}. Therefore, we apply Corollary \ref{egtrweryrthrthhrtdhrtht} and obtain the result for $GC = \min\{0.04,{0.8BC\over 2^{\Psi}}\}$.

\end{proof}
\fi


\subsubsection{Winning the $(\alpha,\beta)$-game}\label{dfjksdfkjsdfkjsdfkjsfd}
To apply Lemma \ref{sdfljsdldjfljeflkjfd} we have to show that there exists a pair $\alpha, \beta$ such that $S_\beta(\alpha)$ has sufficiently many great rows. We show the existence of such a pair by the following lemma.

\newcommand{\existenceLemma}{
Let $k\ge 5$. There exists a pair $\alpha,\beta$ such that the following is true. Denote by $X$ the number of great rows in $S_{\beta}(\alpha)$. Then
$$
\frac{wX2^\beta }{t_{\alpha}} \ge 1.
$$
}

\begin{lemma}\label{sdfsdfsdfsdfnsdvnsdvmnsdmnv}
\existenceLemma
\end{lemma}

The proof of Lemma \ref{sdfsdfsdfsdfnsdvnsdvmnsdmnv} can be found in Appendix \ref{bla}.
Now we are ready to prove the main result of section \ref{sdkjnfkasddksak}.
\begin{theorem}\label{sdfsdfsdfsdfsaldkflkaslkflsda}
If the assumptions in Table \ref{table1} are true then there exists a pair $(\alpha, \beta)$ such that an
$(\alpha,\beta)$-game will return $1$ with probability $0.9$. Furthermore, the player that samples $1$
will beat all winners (different from $1$) of all other $(\alpha',\beta')$-games.
The algorithm works in one pass.
\end{theorem}
\begin{proof}
Lemma \ref{sdfsdfsdfsdfnsdvnsdvmnsdmnv} implies that there exists a pair $\alpha,\beta$ that satisfies the conditions of
Lemma \ref{sdfljsdldjfljeflkjfd}. Thus, for this pair Lemma \ref{sdfljsdldjfljeflkjfd} implies that the heavy element will be a winner in at least one of the teams.
Further, we only consider the first $0.5f_1$ occurrences. Note that the assumptions in Table \ref{table2} imply that
$f_l \le G_k^{1/k} < 0.5f_1$ for any $l>1$.
Our algorithm implies that the estimator never exceeds the real frequency.
Thus, the heavy element will beat all winners in all other games and all other teams.
\end{proof}

\subsection{Modifications to the Game Algorithm}
Here we will modify the game several times without affecting the main claim of correctness. The modified version will allow us to save space.
The modifications will be applied for cases where $\beta \le C\log\log n$, for sufficiently large C. Also, see Observation \ref{wsgfwsergfwergf}.

\subsubsection{Saturated Rows }\label{sdfljsjfsljfdlsjef}
If a single row has too many instances of the heavy element, the counter can grow so large that it takes up too much space. Fortunately, there can only be very few rows that have this property, so they can be safely ignored.
The following modification will take place during the game with parameters $\alpha, \beta$ and during round $\gamma$.
All updates that are larger than $2^{4\gamma}$ will be ignored.
We claim that this modification does not affect the correctness of the algorithm.
Define row $i$ as a $\gamma$-saturated row for $\gamma\ge \beta$ if there exists
a row $j$ such that $f_1(\alpha,j) > 2^{4\gamma}$ and $j-i \le 2^{\gamma}$. The number of rows with $f_1(\alpha,j)> 2^{4\gamma}$ is at most $f_1\over  2^{4\gamma}$. Thus, the number of $\gamma$-saturated rows is at most $f_1\over  2^{3\gamma}$. Summing up for all $\gamma \ge \beta$ we conclude that the number of rows that are saturated for any $\gamma$ is at most $f_1\over  2^{3\beta-1}$. Using the bound on the number of good rows from Corollary \ref{dflkmgvldsfgldsflksdflk} we conclude that the total
number of good rows decreases by a negligible amount that does not affect the correctness.

\subsubsection{Storing Round information in Small Space}\label{dfdsfsdfdsf}
%
We will record each winning player's location within their team.
To do so, we will keep a sequence of length at most $3w \gamma \over 3^{\gamma}$ bits and
interpret it as follows.
Each sequence of $\gamma \log_2 3+1$ bits will indicate the offset in the group of $3^{\gamma}$ players that
could play in one group in round $\gamma$. Since at most one winner exists, $\gamma \log_2 3$ bits is sufficient to represent it. We will use an additional bit to indicate the case when all players are inactive.
Since $\gamma \log_2 3+1 \le 3\gamma$ the bound follows.

Using the above representation, it is possible to support the following operations.
Given a player number and the team number, it is possible to check whether the player is still active.
The counters and other data will be stored per group as a sequence of fixed-length words.
Also, given a team number, it is possible to list all winners and their counters.

\subsubsection{Reservoir Sampling with Small Space}\label{sdfwefwefwefweewf}
We now implement the sampling method for players using small space (see Section \ref{dskfjkjasdskjsdf}).  We can instantiate all counters that are needed to reservoir sample using $O(\beta)$ bits since $|\Gamma| = O(2^{\beta})$ with high probability.
Note that this space does not include the space that is needed to store the actual sample.
Only the implementation of the algorithm, i.e. additional structures, is discussed here.

We will use the simple reservoir sampling algorithm from \cite{reservoir}.
Recall that reservoir sampling can maintain one sample using independent coin flips for each new element of the pool.
The coin bias depends only on the length of the pool and
thus the implementation can be carried out with $O(\beta)$ bits if $|\Gamma| \le 2^{O(\beta)}$. We will guarantee this condition by making a player inactive if $|\Gamma| \ge 100* 2^{\beta}$ for the pool of its sample. Recall that
the sampling procedure is explained in  Lemma \ref{fkljsdfkjlsdflkjdsf}.
It is possible to check that the proof of Lemma \ref{fkljsdfkjlsdflkjdsf} is given for this condition (see $(\ref{dsfklmslkdflsdflsd})$). Thus, the statement of Lemma \ref{fkljsdfkjlsdflkjdsf} holds.


\subsection{Space Complexity: Preliminary Analysis}\label{wdfnkjwdfskjwfdjkwef}
Before we proceed to the final modification of our algorithm, let us analyze the space complexity of the current version. We will argue that only one change is needed to achieve the desired $O(n^{1-2/k})$ bound.

On a high level, our algorithm collects samples, maintains counters and compares them to declare winners.
To implement the algorithm it is necessary to store the information contained in the tables below.
For each triple of $\eta,\alpha, \beta$ we play the game.
The following is the explanation of all data structures that we use with their space complexities.
During each step of the game, our algorithm reads the next element of the stream and identifies
the row of this element in the matrix $M_{\alpha}$.
At any moment we store information about all past rows and the current rows.
We store a structure for each row $i \in [r_{\alpha}]$.
Note that given the number of the current row $j$ and the number of the past row $i$
it is straightforward to determine the round $\gamma$.

In the following table we fix $\alpha, \beta$, and $\gamma$.

\begin{table}[H]
\centering
\begin{tabular}{|l|l|l|}
\hline
Basic Info  & Bits&Explanation\\ \hline
IDs of samples &  $\log n$ &If x is sampled we have to remember x to count its frequency and \\
               & &compare with other samples. $\log n$ bits are needed to store one ID.\\ \hline
Frequency & $O(\gamma)$ &  To store the current frequency. \\
Counters & &For an $(\alpha,\beta)$-game and for round $\gamma\ge \beta$, \\
          & & $O(\gamma)$ bits are sufficient per counter. \\
         & & This is because we ignore all large updates (See Section \ref{sdfljsjfsljfdlsjef}.) \\ \hline
Location of winning &&\\
players in the team  & & \\
 & $O( \gamma)$ &  See Section \ref{dfdsfsdfdsf}. \\ \hline
Id of a row & $O(\log n)$ &  \\ \hline
Hash Function $g$&  $O(\log(n))$& Since $g$ is pairwise independent. See \cite{ams} for details.\\ \hline
Reservoir Counter &  $O(\beta)$& See Section \ref{sdfwefwefwefweewf}\\ \hline
\end{tabular}
\caption{Space Complexity of One Round of the Game}
\end{table}

Each structure will be stored as a sequence of fixed-length words. For each row $i\le j$ where $j$ is the current row we will store the value of $i$. Also, for each group of $3^{\gamma}$ players in the $\gamma$-th round (defined by $j-i$) we will store for each winner its ID, its counter, and its player's number.
Structures per row per game:

\begin{table}[H]
\centering
\begin{tabular}{|l|l|l|}
\hline
Data Structure  & Bits&Explanation\\ \hline
Array of IDs of samples &  $O(\log n{w\over 3^{\gamma}})$ &  \\ \hline
Array of Counters  & & \\ & $O(\gamma{w\over 3^{\gamma}})$ &  \\ \hline
Location of winning players & & \\
 & $O(\gamma{w\over 3^{\gamma}})$&   \\ \hline
Id of a row & $O(\log n)$ &  \\ \hline
Hash Functions&  $O(\log(n))$& One function per row. See Section \ref{dskfjkjasdskjsdf}.  \\ \hline
Reservoir Sampling Instances&  $O(\beta w)$&\\ \hline
\end{tabular}
\caption{Space Complexity Required to Store Each Row of Matrix}
\end{table}

Structures per game (all rows, recall that $\gamma \ge \beta$ and that there exist at most $2^{\gamma}$ rows with current round $\gamma$):

\begin{table}[H]
\centering
\begin{tabular}{|l|l|l|}
\hline
Data Structure  & Bits&Explanation\\ \hline
Arrays of IDs of samples &  $O({w}\log n)$ &  \\ \hline
Arrays of Counters  & & \\ & $O({w})$ &  \\ \hline
Location of winning players & & \\
 & $O({w})$&   \\ \hline
Id of a row & $O((\log n)r_{\alpha})$ &  \\ \hline
Hash Functions&  $O((\log n)r_{\alpha})$&  only for the last $\log n$ rows\\ \hline
Reservoir Sampling Instances&  $O(\beta w)$& only for the current row\\ \hline
\end{tabular}
\caption{Space Complexity of One $(\alpha,\beta)$-Game}
\end{table}

By ignoring all updates that are larger than $2^{4\gamma}$ we only need $O(\gamma)$ bits per counter in round $\gamma$.
Let us summarize the space for frequency counters. Also, we assume that $k>3$ in which case $r_\alpha \log n = o(n^{1-2/k})$ for all values of $\alpha$ in the range of Algorithm \ref{sdjnfsdfsjdfskjd}.
After the $\gamma$-th round every group of $3^\gamma$ players has at most one active member.
There are at most $2^\gamma$ teams between rounds $\gamma$ and $\gamma+1$.
It follows that the number of active players is
$O(w\sum_{\gamma=1}^{\lceil\log(q)\rceil} 2^\gamma3^{-\gamma}\gamma)$.
This is a converging sequence and thus, the total space is $O(w)$.
Similar arguments can be made for other structures.

Summing up for all $\beta,\alpha$ and noting that $w = {n^{1-2/k}\over 2^{\mu\beta}}$ we observe\footnote{Where $\mu$ is an absolute constant defined in Table \ref{table2}.} that
the arguments in this section imply that the following observation is true:
\begin{lemma}
If we could reduce the space for storing the IDs
from $\log n$ to $O(\gamma)$ then the algorithm will work with $O(n^{1-2/k})$ bits.
\end{lemma}

In the next section we will do precisely that: reduce the space for storing the IDs of each element sampled.
Before we proceed let us justify the fact that the modifications are necessary only for small $\beta$.
In the beginning of Section \ref{wdfnkjwdfskjwfdjkwef} we assumed that $\beta = O(\log\log n)$.
In the next observation we show that otherwise the problem can be solved without any modifications.

%
%

\begin{observation}\label{wsgfwsergfwergf}
The problem is solved without any modification when the correct value of $\beta$ is larger than $C\log\log n$, for sufficiently large C.
Indeed, the sampling complexity decreases exponentially with $\beta$.
Therefore, for any $\beta > C\log\log n$, again for sufficiently large C, the total space complexity is $O(n^{1-2/k})$ even if we use $\log n$ bits for counters and IDs.
\end{observation}

\subsection{Signatures Instead of IDs}\label{dfsdfhssfhstrt}

\subsubsection{Concept}
As shown in the previous section, the Game algorithm does not provide an improvement in space complexity \cite{DBLP:journals/corr/abs-1212-0202}. In this section we describe the process of modifying the Game to use hash based signatures to store which elements have been sampled instead of storing the ID of the element. An intuitive explanation of this approach is provided in section \ref{adfsdfsefsdfsdf}.

\subsubsection{Signature Assignment Algorithm}
Let $s = O(\log(n))$. Let $sig$ be an $n\times s$ matrix with i.i.d. columns. Each column is a vector
with uniform zero-one entries that are $4$-wise independent. It follows that we need $O(\log^2(n))$ bits to represent $sig$. Denote by $R_j(i)$  the first $j$ bits in the $i$-th row of $sig$, i.e., $R_j(i) = \{sig_{i,1},\dots, sig_{i,j}\}$.
Let $\varrho > 100$ be a constant.

Recall that after a team is sampled from its initial row, it continues to play on every subsequent row until the end of the matrix. For each team, split the rounds they play into three phases. Phase one begins with the first round and ends after the $\lceil\log\log n\rceil$-th round.
Phase two starts right after phase one ends and ends after the $\lceil10\log\log n\rceil$-th round ends.
Phase three starts right after phase two ends and continues until the end of the game.

When a player samples an element, the ID of the element is not represented explicitly using $\log(n)$ bits.
Instead, initially the sample is assigned a signature $R_{\varrho\beta}(p)$ and this value is given to the player to count.  During the $\gamma$-th round in the first phase, the sample is represented by a signature $R_{\varrho\gamma}(p)$.
The signatures are extended as follows. If a player with a given signature has a counter that grows large enough, the next time we see an element with matching signature, we add more bits of identification based on the element seen to increase the resolution in counting future elements for that player.

The counter is incremented every time the hash function gives the right value and the signatures of the new stream element matches the stored signature for the sample. Thus, collisions are possible during the first round and the counters may be incorrect.

During the second phase we assign an ID to the player. Once the second phase is reached, the next time an element is read that has the same signature as the player, the ID of that element is given to the player.
After assigning the ID two different counters are kept. One counter is the old counter; its value does not change. Another counter is a counter that starts with the value of zero and counts the number of elements that have the same ID as the player.
The value of the counter during the second phase is the sum of the values of the old and new counters.
Thus, during the second phase the counters still can be incorrect, but no new counting errors will be introduced.

Finally, during the third phase we discard the old counter and use only the new counter.
In phase three we underestimate the original counter of the algorithm.

\begin{theorem}\label{sigok}
Theorem \ref{sdfsdfsdfsdfsaldkflkaslkflsda} is still correct after the modifications performed in Section \ref{dfsdfhssfhstrt}
\end{theorem}
The proof of Theorem \ref{sigok} can be found in Appendix \ref{signaturesproofs}.

\subsection{Proving Theorem \ref{wdfkjnwkejfkjwefkjwef}: Finding Heavy Hitters}

We now prove Theorem \ref{wdfkjnwkejfkjwefkjwef} from section \ref{Main Result}.

\begin{proof}
First we will work under the assumptions from Table \ref{table1} for $f_1$ and show we can find a heavy element with constant probability. Our analysis will be true for any heavy element, not necessarily $1$. The correctness follows from Theorem \ref{sigok}. Let us bound the space complexity of our algorithm. As we discuss in Section \ref{wdfnkjwdfskjwfdjkwef} we need to demonstrate that  $O(w\gamma)$ bits are sufficient to store the IDs of all players. This is indeed the case; as we show in Section \ref{dfsdfhssfhstrt}, in the first phase we need $3w \varrho\gamma$ bits. During all other phases we need only $o(w)$ bits.

Thus, for fixed $\beta$ the space complexity of playing an ($\alpha,\beta$)-game is $O(w)$. Recall the value of $w$ from Table \ref{table2} implies that the space complexity is $O({n^{1-2/k}\over 2^{\mu\beta}})$. Now, let us compute the space complexity of playing all games in parallel. We play this game for all $\eta \in [0,RANGE]$. For a fixed value of $\eta$ we play games on a range of parameters ($\alpha,\beta$) that are given in Algorithm \ref{sdjnfsdfsjdfskjd} (see also Section \ref{F}).
%
Summing up for all $\eta, \alpha, $ and $\beta$ we obtain a geometric series that sums to $O(n^{1-2/k})$.
Indeed, fix $\eta$ and observe that for the fixed $\eta$ the cost of all games played has the following upper bound.
\begin{equation}\label{ljfsjdfjsdfjs}
n^{1-2/k}\left(\sum_{u=1.5\eta}^{RANGE} 2^{-0.8\mu u}\right)
\end{equation}
Indeed, $\beta\ge 0.8u$ in all games and therefore the above bound is true.
The value in $(\ref{ljfsjdfjsdfjs})$ is further bounded by $\frac{n^{1-2/k}}{ 2^{C\eta}}$ for some absolute constant $\eta$.
Summing over all $\eta \in \{0,\dots, RANGE\}$ we conclude that the upper bound on the total cost of all games is $O(n^{1-2/k})$.

Therefore, the following statement is true, under the assumptions from Table \ref{table1}. Algorithm \ref{sdjnfsdfsjdfskjd} finds the heavy element, if one exists, uses $O(n^{1-2/k})$ bits,
and works correctly with a constant probability.

To finish the proof of Theorem \ref{wdfkjnwkejfkjwefkjwef}
let us now remove the assumptions from Table \ref{table1}. In the first pass we compute the ratio $F_0\over F_1$.
In the second pass we sample the stream with probability $F_0\over F_1$ (see Section \ref{edfdsfsdfadsfasdfadsf}).
During the second pass we will find some indices $i$ that will contain indices of the heavy elements
with high probability. During the third pass we will compute the values $f_i$ precisely.

So far, we have shown how to find heavy hitters given the assumption on the range of $f_1$ from Table \ref{table1}. To remove this assumption, for parameter $\rho$ it remains to be proven that if $f_i^k \ge \rho F_k$ then all such indices will be output with probability at least $1-\delta$
and that the total cost is as stated.
To show the first part, let $h [n]\mapsto [z]$ where $z=O({1\over \rho^2})$.
Let $x\in [z]$ and denote by $D_{x}$ the substream of $D$ defined as
$D_{x} = \{p_i: h(p_i) = x\}$.
The Markov inequality and union bound imply the following statement.
For sufficiently large $z$, with probability $0.01$ for all $\rho$-heavy elements the assumptions of the Game are correct.
Therefore the space complexity of finding $\rho$-heavy elements as stated in Theorem \ref{wdfkjnwkejfkjwefkjwef} becomes $O(\frac{1}{\rho^C}n^{1-2/k})$.  In a similar way, we can show the same bound by replacing $n$ with $F_0$.
\end{proof}

\subsection{Removing Assumptions}\label{sdflklkjsdfjldsjf}
\subsubsection{Values of $\eta$}\label{sdgsdfsgdgfsd}
Assuming that $F_1 = O(n)$ we show how to address the case when $\eta > 3\log\log n$.
If $\eta > 3\log\log n$ then $F_k\ge \log^3 n$. Recall that result from the pick-and-drop sampling
is that in this case we can apply the sampling using $O(n^{1-2/k}\over \log n)$ samples.
In this case the total space is at most $O(n^{1-2/k})$ bits.
When $\eta$ is odd we will consider $\eta' = \eta-1$ and repeat the analysis with the appropriate change in constants.

\subsubsection{$1$ does not expire}\label{wefljlnkewjfkjewfkjkwef}
In this section we will remove the assumption that $1$ does not expire (See Table \ref{table1}).
Recall that the $(\alpha, \beta)$-game is played on matrix $M_\alpha$.
Note that there is a one-to-one correspondence between the elements of the stream and the entries of matrix $M_\alpha$.
Suppose that $s$ corresponds to an entry $i,j$ in the matrix $M_\alpha$, specifically $s = (i-1)t_{\alpha}+j$.
If $p_s$ has been sampled during the ($\alpha,\beta$)-game then the value of the counter after the
$\gamma$-th round is at least $\sum_{l=i+1}^{i+2^\gamma} f_1(\alpha,l)$. If we can show that this value is larger than $2^\gamma$ then the player will not expire at the $\gamma$-th round. Below we will show that there are many elements $p_s$ such that $p_s=1$ and such that they will not expire at any ($\alpha,\beta$)-game.
Next we will be using results\footnote{See \cite{DBLP:journals/corr/abs-1212-0202} or Appendix \ref{asdfadsfasfasf} for additional details and proofs.} from \cite{DBLP:journals/corr/abs-1212-0202}.

\begin{definition}\label{def: dkjfjrglrlttrg1}
Let $U = \{u_1,\dots,u_M\}$ and $W=\{w_1,\dots,w_M\}$ be two
sequences of non-negative integers. Let $(i,j)$ be a pair such that $1\le i\le
M$ and $1\le j\le u_i$. Denote $(i,j)$ as a \emph{losing} pair (w.r.t. sequences $U,W$) if there exists $h, i\le h\le M$ such that:
\begin{equation}\label{sdlnlsdljsdjlskd}
-j + \sum_{l=i}^h(u_l-w_l) < 0.
\end{equation}
Denote any pair that is not a losing pair as a \emph{winning} pair.
\end{definition}
The following is Lemma $2.20$ in \cite{DBLP:journals/corr/abs-1212-0202}. For completeness, we provide the proof in Appendix \ref{kljdflkjdsfsfkldfknkdffknsdkksd}.
\begin{lemma}\label{fct:erergerger1}
If $\sum_{s=1}^t(u_s-w_s) > 0$ then there exist at least $\sum_{s=1}^t(u_s-w_s)$ winning pairs.
\end{lemma}

\begin{definition}\label{dsfdsfsdffsdgsdgsdgdfs}
Denote two sequences $V = \{v_1,\dots, v_{r_0}\}$ and $W= \{w_1,\dots, w_{r_0}\}$ as follows:
\begin{equation}\label{cxncxnnxcnxcxncc}
v_i = f_1(0,i), \ \ \ w_i = 2^{\eta}.
\end{equation}
By Lemma \ref{fct:erergerger1} there exist at least
\begin{equation}\label{cxncxnnxcnxcxnc}
\sum_{i=1}^{r_0} (v_i - w_i) = f_1 - 2^{\eta}r_0
\end{equation}
winning pairs $(i,j)$ w.r.t. $(V,W)$ (see Definition \ref{def: dkjfjrglrlttrg1}).
The equality follows since $\sum_{i=1}^{r_0} v_i = f_1$ and $\sum_{i=1}^{r_0} w_i = 2^{\eta}r_0$.
Note that there is a injection of the set of the winning pairs to the set of appearances of $1$ in the stream. To see that, consider a winning pair $(i,j)$ where $1\le i\le r_0$ and $1\le j\le v_i$.

Let $J$ be the column in the $i$-th row of $M_0$ where $1$ appears for the $j$-th time (in the $i$-th row).
Let $s = (i-1)t_0+J$. It is not hard to see that $p_s = 1$ and that two distinct winning pairs produce two distinct elements of the stream.
Let $STEADY$ be a set of all such elements of the stream.
\end{definition}

Below we will show that steady elements do not expire during the games.
In the reminder of the paper we will restrict our analysis to steady elements.
Our assumptions on $f_1$ from Table \ref{table1} imply that $f_1 - 2^{\eta}r_0 > 0.99f_1$.
Therefore a constant fraction of all appearances will not be discarded.

\begin{fact}\label{wdkjwefkjwe}
A steady element will not expire during a $(0,\beta)$-game.
\end{fact}
\ifNOTtalk
\begin{proof}
Let $p_s$ be a steady element. It is sufficient to show that if a player samples $p_s$ then for any $\gamma$ the counter of the player after reading $2^\gamma$ rows is at least $TR$.

Fix $\gamma > 0$ and let $(i,j)$ be the corresponding entry of $M_0$ from Definition \ref{dsfdsfsdffsdgsdgsdgdfs}.
In this case it is sufficient to show that
\begin{equation}\label{dlfndsljfjsdfjksdkjsd}
\sum_{l=i+1}^{i+2^\gamma}f_1(0,i) \ge TR.
\end{equation}

This is indeed true because $(i,j)$ is a winning pair in the sequence $(V,W)$. Thus, by the definition of a winning pair we have
\begin{equation}\label{dlfndsljfjsdfjksdkjsd1}
\sum_{l=i}^{i+2^\gamma}v_l-j \ge \sum_{l=i}^{i+2^\gamma}w_l.
\end{equation}

Next recall that $w_l=2^\eta$ and $v_l = f_1(0,l)$. Thus $(\ref{dlfndsljfjsdfjksdkjsd1})$ becomes
\begin{equation}\label{dlfndsljfjsdfjksdkjsd2}
\sum_{l=i}^{i+2^\gamma}f_1(0,l)-j \ge 2^{\eta+\gamma}.
\end{equation}
Recall  $TR \le 2^{\eta+\gamma}$ for $\alpha=0$ (this follows from the definition of $TR$ in Table \ref{table1}.)
Also by the definition $j\le f_1(0,i)$ and thus $(\ref{dlfndsljfjsdfjksdkjsd2})$ implies $(\ref{dlfndsljfjsdfjksdkjsd})$.

%
%
%
\end{proof}

\begin{fact}\label{wdkjwefkjwe}
Let $\alpha<0$.
A steady element will not expire during an $(\alpha,\beta)$-game.
\end{fact}
\begin{proof}
Note that each row of $M_{\alpha}$ corresponds to $2^{-\alpha}$ rows in $M_0$.
Therefore $2^\gamma$ rows in $M_\alpha$ correspond to the $2^{-\alpha}$ rows in $M_0$.
Since $(i,j)$ is a winning pair, we obtain that after reading any $2^\gamma$ rows in $M_\alpha$ the counter is at least
$2^{\eta + \gamma-\alpha}$.

Specifically, if $p_s$ is steady and $(i,j)$ is the corresponding entry in $M_\alpha$ then after reading
rows $i+1,\dots,i+2^\gamma$ in $M_\alpha$ the counter will be at least $2^{\eta + \gamma-\alpha}$.
The fact follows since $TR \le 2^{\eta + \gamma-\alpha}$.

%
\end{proof}

\begin{fact}\label{dsjsdkkjlsdkjsd}
Let $\alpha>0$.
A steady element will not expire during an $(\alpha,\beta)$-game.
\end{fact}
\ifNOTtalk
\begin{proof}
Note that each row of $M_{0}$ corresponds to $2^{\alpha}$ rows in $M_0$.
Therefore reading $2^\gamma$ rows in $M_\alpha$ must include reading of at least $2^{\gamma-\alpha} -1$ full rows in $M_0$ (we do not count at most $2^alpha$ rows in $M_\alpha$ that correspond to the first partial row in $M_0$).
Since we assume that $\gamma>\alpha$ we conclude that the number of rows that will be read in $M_0$ is at least $2^{\gamma-\alpha-1} = TR$.

%
%
%
\end{proof}
\fi

\subsubsection{$F_1 \le C_2n$.}\label{edfdsfsdfadsfasdfadsf}
In this section we will remove the assumption that $F_1 \le C_2n$.
Instead we will show that the second pass of this algorithm allows us to disregard this assumption.
In the first pass we will compute $F_1$ and approximate $F_0$. In the second pass we will subsample the stream using $p={n/F_1}$.
The length of the sample stream is at most $10n$ w.p. $0.9$, so we will choose $C_2 \ge 10$.

See Appendix \ref{asdfadsfasfasf} for the proofs.



\newpage

\section{Martingale Sketches}\label{dsjkfkjsdfkjdsfkjsdfkj}

\subsection{Introduction}
Having established a streaming algorithm which can efficiently compute the heavy hitters of a stream, we present a reduction of the problem of frequency moment approximation to that of finding heavy hitters. In general, this analysis will show that the problem of approximating the sum of am implicit vector is the same problem as finding the heavy elements of that vector, up to a constant factor. As a direct corollary, and using our new heavy hitter algorithm, we obtain a new lowest bound on space required for this problem. The analysis that follows is independent of the previous analysis, and as such the notation, constants, and variables herein do not carry over their previous meanings.

\subsection{Preliminaries}

Recall that a sequence of random variables with finite mean $B = \{b_0,\dots, b_t\}$ is a martingale if for all $i=1,\dots, t:$
\begin{equation}\label{dwdwsf}
E(b_i|b_{i-1}, \dots, b_0) = b_{i-1}.
\end{equation}
Without loss of generality we will assume that $b_0$ is a fixed\footnote{If this is not the case then we can add $b_{-1} = E(b_0)$ and define a new martingale $b'_i = b_{i-1}$.} number. Note that for any $i$:
\begin{equation}\label{dwdwfdfgdsf}
E(b_i) = b_0.
\end{equation}

\begin{definition}
Let $Z=\{Z_1,\dots, Z_t\}$ be a sequence of random variables and $B = \{b_0,\dots, b_t\}$ be a martingale.
Let $u,\epsilon \in (0,1)$ be parameters.
We say that $Z$ is an $(\epsilon, u)$-\textbf{fixing sequence}\footnote{It is possible to replace $0.1$ with any constant $c<1$ and obtain similar results.} with respect to $B$ if for all $i=0\dots t-1$:
\begin{equation}\label{dwdffwsf}
P(|Z_{i+1}+b_{i+1} - b_{i}| \ge \epsilon_i b_i) \le 0.1(1-u)u^i,
\end{equation}
where
\begin{equation}\label{dfdfwdwffdsf}
 \epsilon_i = (0.1(1-u)u^i)\epsilon.
\end{equation}
\end{definition}

\begin{lemma}\label{sdjkkjhdsf}(Fixing Lemma)
Let $Z$ be an $(\epsilon, u)$-fixing sequence w.r.t. martingale $B$.
Define $S= b_t + \sum_{i=1}^t Z_i.$
Then
\begin{equation}\label{dfrerdfwdwsf}
 P(|S - b_0|\ge \epsilon b_0) \le 0.2.
\end{equation}
\end{lemma}
\ifNOTtalk
\begin{proof}
Note that $b_t = b_0 + \sum_{i=0}^{t-1} (b_{i+1}-b_{i})$ and thus by definition of $S$:
$$
S = b_t + \sum_{i=1}^t Z_i =  b_0 + \sum_{i=0}^{t-1} (b_{i+1}-b_{i}+Z_{i+1}).
$$
Therefore $S - b_0 = \sum_{i=0}^{t-1} (b_{i+1}-b_{i}+Z_{i+1})$ and
\begin{equation}\label{dfrsddserdfwsdsdfdwsf}
|S-b_0| = |\sum_{i=0}^{t-1} (b_{i+1}-b_{i}+Z_{i+1})| \le \sum_{i=0}^{t-1} |b_{i+1}-b_{i}+Z_{i+1}|,
\end{equation}
\begin{equation}\label{dfrsdddsrdfwddsdswsf}
P(|S-b_0|\ge \epsilon b_0)\le P(\sum_{i=0}^{t-1} |b_{i+1}-b_{i}+Z_{i+1}|\ge \epsilon b_0).
\end{equation}

Let $\{X_i\}_{i=0}^{t-1}$ and $\{Y_i\}_{i=0}^{t-1}$ be two sequences of random variables.
If $\sum_{i=0}^{t-1} X_i \ge \sum_{i=0}^{t-1} Y_i$ then there exists at least one $i$ such that $X_i \ge Y_i$.
Therefore, $P(\sum_{i=0}^{t-1} X_i \ge \sum_{i=0}^{t-1} Y_i) \le P(\cup_{i=0}^{t-1}(X_i \ge Y_i))$. Applying the inequality with $X_i = |b_{i+1}-b_{i}+Z_{i+1}|$ and $Y_i = \epsilon_i b_i$ we obtain:
\begin{equation}\label{dfdserdfwdwsf}
P(\sum_{i=0}^{t-1} |b_{i+1}-b_{i}+Z_{i+1}|\ge \sum_{i=0}^{t-1} \epsilon_i b_i)\le P(\bigcup_{i=0}^{t-1}(|b_{i+1}-b_{i}+Z_{i+1}|\ge \epsilon_i b_i))
\end{equation}
By union bound:
\begin{equation}\label{dwsf}
P(\bigcup_{i=0}^{t-1}(|b_{i+1}-b_{i}+Z_{i+1}|\ge \epsilon_i b_i)) \le \sum_{i=0}^{t-1}P(|b_{i+1}-b_{i}+Z_{i+1}|\ge \epsilon_i b_i).
\end{equation}
By applying $(\ref{dwdffwsf})$:
\begin{equation}\label{dwdfdfsf}
\sum_{i=0}^{t-1}P(|b_{i+1}-b_{i}+Z_{i+1}|\ge \epsilon_i b_i) \le 0.1(1-u)\sum_{i=0}^{t-1} u^i = 0.1(1-u^t),
\end{equation}
and since $0 < u < 1$,
\begin{equation}\label{ddeflkewlrflkerf}
0.1(1-u^t) \le 0.1.
\end{equation}
Thus, by $(\ref{dfdserdfwdwsf}),(\ref{dwsf}), (\ref{dwdfdfsf})$ and $(\ref{ddeflkewlrflkerf})$:
\begin{equation}\label{dwsfddssd}
P(\sum_{i=0}^{t-1} |b_{i+1}-b_{i}+Z_{i+1}|\ge \sum_{i=0}^{t-1} \epsilon_i b_i) \le 0.1.
\end{equation}
Let us bound the total error:
\begin{equation}\label{ddssdwsfddssd}
E(\sum_{i=0}^{t-1} \epsilon_i b_i)= \sum_{i=0}^{t-1} \epsilon_i b_0  = \epsilon b_0\sum_{i=0}^{t-1} 0.1(1-u)u^i \le 0.1\epsilon b_0.
\end{equation}
Here the first equality follows from $(\ref{dwdwfdfgdsf})$, the second equation follows from $(\ref{dfdfwdwffdsf})$ and the last inequality
follows from $(\ref{ddeflkewlrflkerf})$.
Thus, by the Markov inequality:
\begin{equation}\label{ddsdfdsdwsffddffdfdddssd}
P(\sum_{i=0}^{t-1} \epsilon_i b_i \ge \epsilon b_0) \le 0.1.
\end{equation}

Let $X,C,A$ be three random variable. Since $(X < C) \cap (C < A) \implies X < A$, by the contrapositive, $X \ge A \implies (X \ge C) \cup (C \ge A)$. Thus
\begin{equation}\label{dfljjnjdfgljdfdglkjfdgljdfd}
P(X \ge A) \le P(X \ge C) + P(C \ge A)
\end{equation}
Let us apply $(\ref{dfljjnjdfgljdfdglkjfdgljdfd})$ with the following variables
$X = \{\sum_{i=0}^{t-1}|b_{i+1} - b_i + Z_{i+1}|\}, \ A = \epsilon b_0, C = \left(\sum_{i=0}^{t-1} \epsilon_i b_i\right)$.
\begin{equation}\label{ljlerfljdsjgjegkj}
P(\sum_{i=0}^{t-1}|b_{i+1} - b_i + Z_{i+1}| \ge \epsilon b_0) \le P(\sum_{i=0}^{t-1}|b_{i+1} - b_i + Z_{i+1}| \ge \left(\sum_{i=0}^{t-1} \epsilon_i b_i\right)) + P(\left(\sum_{i=0}^{t-1} \epsilon_i b_i\right) \ge \epsilon b_0)\le 0.2.
\end{equation}
Here the first inequality follows from $(\ref{dfljjnjdfgljdfdglkjfdgljdfd})$ and the last inequality follows from $(\ref{dwsfddssd})$ and  $(\ref{ddsdfdsdwsffddffdfdddssd})$.
Combining $(\ref{ljlerfljdsjgjegkj})$ with $(\ref{dfrsdddsrdfwddsdswsf})$ gives the result.

\end{proof}
\fi

\begin{definition}\label{weewrfwefwe}
Let $C=\{c_0,\dots, c_t\}$ be a sequence of random variables and let $u\in (0,1)$ be a constant.
We say that $C$ is a $u$-\textbf{geometric} sequence if for all $i=1,\dots, t$:
\begin{equation}\label{dwddfdfgcfdsvsdffwsf}
P(c_i \ge u^ic_0) \le 0.1(1-u)u^i.
\end{equation}
\end{definition}

\begin{lemma}\label{sdjkkjfsdssfdhdffdsf}
Let $C$ be a $u$-\textbf{geometric} sequence.
Then
\begin{equation}\label{dfrsdffddffdfderdfwdwsf}
 P(\sum_{i=1}^tc_i \ge {c_0\over 1-u}) \le 0.1.
\end{equation}
\end{lemma}
\ifNOTtalk
\begin{proof}
Consider three events $A = (\cap_{i=1}^t (c_i < u^ic_0))$, $B=(\sum_{i=1}^t c_i < c_0\sum_{i=1}^t u^i )$ and $C=(\sum_{i=1}^t c_i < c_0(1-u)^{-1})$.
Then
\begin{equation}\label{dwddfdssdsdddfgffwsf}
A \subseteq B \subseteq C.
\end{equation}
Thus, $P(\bar{C}) \le P(\bar{A})$:
\begin{equation}\label{dwdddsfdsfdffdssdsdfgfdddfgffwsf}
P(\sum_{i=1}^t c_i \ge c_0(1-u)^{-1}) \le P(\bigcup_{i=1}^t (c_i \ge u^ic_0)),
\end{equation}
By union bound and $(\ref{dwddfdfgcfdsvsdffwsf})$:
\begin{equation}\label{dwdddsfdsfdfsdfdffgfdssdsdddfgffwsf}
P(\bigcup_{i=1}^t (c_i \ge u^ic_0)) \le  \sum_{i=1}^t P(c_i \ge u^ic_0) \le 0.1(1-u)\sum_{i=1}^tu^i\le 0.1.
\end{equation}
Combining $(\ref{dwdddsfdsfdffdssdsdfgfdddfgffwsf})$ and $(\ref{dwdddsfdsfdfsdfdffgfdssdsdddfgffwsf})$ we obtain the result.
\end{proof}
\fi

\begin{theorem}(Fixing Theorem)\label{kjdkhdhsa}
Let $B$ be a martingale sequence of random variables and let $Z$ be an $(\epsilon,u)$-fixing sequence w.r.t. $B$. Suppose that the value of $b_0$ is unknown, but we are given $b_t$. Also, suppose that $c_j$ is the cost of computing $Z_j$ for $j \in \{1, 2, \ldots, t\}$.
Let $c_0$ be a parameter and let $C = \{c_0, c_1, \ldots, c_t\}$.

If $C$ is $u$-geometric, then the following statement is true w.p. at least $0.7$.
The value of $b_t + \sum_{i=1}^t Z_i$ is a $(1\pm \epsilon)$-approximation of $b_0$ and the total cost of $Z$ is at most $(1-u)^{-1}c_0$.
\end{theorem}
\begin{proof}
We will apply Lemmas \ref{sdjkkjhdsf} and \ref{sdjkkjfsdssfdhdffdsf}.
The total cost of computing $Z$ is $\sum_{i=1}^t c_i$.
Denote $S= b_t + \sum_{i=1}^t Z_i.$ By $(\ref{dfrsdffddffdfderdfwdwsf})$ and $(\ref{dfrerdfwdwsf})$:
$$
 P((|S - b_0|\ge \epsilon b_0) \cup (\sum_{i=1}^t c_i \ge (1-u)^{-1}c_0)) \le 0.2 + 0.1 = 0.3.
$$
\end{proof}

\subsection{Reduction from $L_1$ to Heavy Elements}

\begin{definition}\label{deddsfdsfsdf}
Let $t = O(\log n)$.  Let $H$ be a $t\times n$ matrix with the random entries $h_{i,j}$ satisfying the following properties
\begin{enumerate}
\item $h_{i,j} \in \{0,1\}$,
\item $E(h_{i,j}) = 0.5$,
\item $\{h_{i,j}\}_{j=1}^n$ are pairwise independent (for a fixed $i$),
\item The rows of $H$ are independent.
\end{enumerate}
\end{definition}

We call such $H$ a \textbf{domain-sampling} matrix\footnote{Note that such matrices (or essentially equivalent) have been widely used in streaming and related areas.}. Let $\beta$ be a fixed vector, $\beta = \{\beta_1,\dots, \beta_n\}$.
Define matrix $V$ as follows.
\begin{equation}\label{sampleone}
v_{1,j} = 2\beta_jh_{1,j},
\end{equation}
\begin{equation}\label{sampletwo}
v_{i,j} = 2v_{i-1, j}h_{i,j}  \text{ for } i > 1.
\end{equation}

Let $V_i$ and $H_i$  be the $i$-th rows of $V$ and $H$ respectfully.
In other words, $V_i$ represents the entries of $\beta$ that have been sampled (and scaled up by a factor of $2^i$), with non-sampled entries replaced with $0$.
In particular, the following property follows from the definitions:
\begin{equation}\label{weffljwelkjrrwkjef}
v_{i,j} = 2^{i}\left(\prod_{l=1}^{i} h_{l,j} \right)\beta_j.
\end{equation}
Define new random variables
\begin{equation}
b_0 = |\beta|
\end{equation}
and
\begin{equation}\label{sderertertretfsdsdg}
b_{i} = |V_i|.
\end{equation}

Direct computations imply that
\begin{equation}\label{gfbjgfdbljfdlgjbghlj}
b_{i+1} = 2\sum_{i=1}^n h_{i+1,j}v_{i,j},
\end{equation}
or, more succinctly:
\begin{equation}\label{gfbjgfdbljfdlgtythyhjbghlj}
b_{i+1} = 2\langle V_i, H_{i+1}\rangle,
\end{equation}
where $\langle X,Y\rangle$ indicates the inner product of vectors $X$ and $Y$.

It follows that $B = \{b_0,\dots, b_t\}$ is a martingale.
\begin{fact}\label{wfnldjfljdfljwdf}
$B = \{b_0,\dots, b_t\}$ is a martingale.
\end{fact}
\ifNOTtalk
\begin{proof}

First observe that each random variable $b_i$ has finite range, and, therefore,
finite mean. We have to show that $E(b_i| b_{i-1},\dots,b_{0}) = b_{i-1}$ for any $i\ge 1$.
Below we consider the case $i>1$. The case $i=1$ is similar.

For real numbers $a_0,\dots, a_{i-1}$ denote $\Psi = (b_s=a_s, s=0,\dots,i-1)$.
It follows from the definition of conditional expectation that
$$
E(b_i| b_{i-1},\dots, b_0)=  \sum_{\Psi}\frac{E(b_i\one_{\Psi})}{P(\Psi)}\one_{\Psi}\,,
$$
where the sum is taken over all events $\Psi,\,P(\Psi)>0$. (It is clear that the set of such $\Psi$ is finite.)
So, it is enough to show that $E(b_i\one _{\Psi})= a_{i-1}P(\Psi)$ for any such $\Psi$.

We have, according to $(\ref{gfbjgfdbljfdlgjbghlj})$,
$$
E(b_i\one _{\Psi})= 2\sum_{j=1}^n E(h_{i,j}v_{i-1,j}\one_{\Psi}).
$$
According to (4) from Definition \ref{deddsfdsfsdf} and $(\ref{weffljwelkjrrwkjef})$, the random variables $h_{i,j}, 1\le j\le n$, are independent of
$v_{i-1}\one_{\Psi}$. So,
$$
E(h_{i,j}v_{i-1,j}\one_{\Psi}) = E(h_{i,j})E(v_{i-1,j}\one_{\Psi})= \frac12 E(v_{i-1,j}\one_{\Psi})\,.
$$
From here
$$
E(b_i\one _{\Psi})= \sum_{j=1}^n E(v_{i-1,j}\one_{\Psi}) = E\left(\sum_{j=1}^n v_{i-1,j}\one_{\Psi} \right) =
E\left[\left(\sum_{j=1}^n v_{i-1,j}\right)\one_{\Psi} \right]\,.
$$
But,  according to $(\ref{sderertertretfsdsdg})$,  $\sum_{j=1}^n v_{i-1,j} = |V_{i-1}|= b_{i-1}= a_{i-1}$ on the event $\Psi$. Therefore,
$$
E(b_i\one _{\Psi})= a_{i-1}P(\Psi)\,.
$$

\end{proof}
\fi

Define the fixing sequence w.r.t. $B$ as follows.
Let $\epsilon, u$ be parameters (we will define $u$ in the next section).
\noindent
Define
\begin{equation}\label{sdfsdsdg}
\alpha_i = 0.1\epsilon_i^2 u^i(1-u),
\end{equation}
where we use previously defined $(\ref{dfdfwdwffdsf})$
$$
 \epsilon_i = (0.1(1-u)u^i)\epsilon.
$$
Let $S_i$ be a (possibly empty) subset of $\{1,2,\dots, n\}$ such that:
\begin{equation}\label{s set def}
\{j: v_{i,j} \ge \alpha_i |V_i|\} \subseteq S_i.
\end{equation}
Also let $S_0$ be defined similarly with
\begin{equation}
\{j: \beta_j \ge \alpha_0 |\beta|\} \subseteq S_0
\end{equation}
Finally, let us define our sequence $Z = \{Z_1,\dots, Z_t\}$ as:
\begin{equation}\label{adfjsdfkjsdlfkj}
Z_1 = \sum_{j\in S_1} (1-2h_{i,j})\beta_j\\
\end{equation}
\begin{equation}\label{sderertertretfsdsdeeefwwrwg}
Z_{i+1} = \sum_{j\in S_i} (1-2h_{i+1,j})v_{i,j}.
\end{equation}

\begin{lemma}\label{fixingsequence}
$Z$ is a fixing sequence with respect to $B$.
\end{lemma}
\ifNOTtalk
\begin{proof}
Before proving the lemma let us make the following observation.
Let $Y \in R^{n}$ be a (random or fixed) vector. Let $\alpha, \epsilon$ be parameters independent of $Y$. Let $S$ be a (possibly empty) subset of $\{1,2,\dots, n\}$ such that $\{j: y_j \ge \alpha |Y|\} \subseteq S$.
Let $Q \in \{0,1\}^n$ be a random vector with  pairwise independent entries $q_i$ such that $E(q_i) = 0.5$ and such that $Q$ is independent of $Y$ and $S$. Consider
\begin{equation}\label{dfsdefdsfdsfdssdsfsdfsdf}
X = \sum_{j\in S} y_j + 2\sum_{j\notin S} q_jy_j.
\end{equation}
Let $Y$ be fixed. Then $E(X) = |Y|$ and
$$
Var(X) = 4\sum_{i\notin S} y^2_iVar(q_i) = \sum_{i\notin S} y_i^2
$$
$$
\sum_{i\notin S} y_i^2 \le \sum_{i\notin S}\alpha|Y|y_i = \alpha|Y|\sum_{i\notin S}y_i \le \alpha|Y|^2
$$
Thus, by Chebyshev inequality:
\begin{equation}\label{dfsdsfsdfsdf}
P(\big|X - |Y|\big| \ge \epsilon|Y|) \le {\alpha \over \epsilon^2}.
\end{equation}
Integrating over $R^n$ with respect to the distribution of $Y$ we obtain $(\ref{dfsdefdsfdsfdssdsfsdfsdf})$ for random $Y$.

By definition $(\ref{dwdffwsf})$, we have to show that
$$
P(|Z_{i+1}+b_{i+1} - b_{i}| \ge \epsilon_i b_i) \le 0.1(1-u)u^i,
$$
We begin by observing that:
\begin{equation}\label{dwdffwgfhfghfghsf}
Z_{i+1} + b_{i+1} =  \sum_{j\in S_i} (1-2h_{i+1,j})v_{i,j} + \sum_{j=1}^n 2h_{i+1,j}v_{i,j},
\end{equation}
where we use $(\ref{sderertertretfsdsdeeefwwrwg})$ and $(\ref{gfbjgfdbljfdlgjbghlj})$.
After regrouping we obtain:
\begin{equation}\label{dwdffwrtretrettyrgfhfghfghsf}
Z_{i+1} + b_{i+1} = \sum_{j\in S_i} v_{i,j} + 2\sum_{j\notin S_i} h_{i+1,j}v_{i,j}.
\end{equation}
Our goal now is to apply the observation $(\ref{dfsdsfsdfsdf})$. Indeed, let us substitute in the above settings
\begin{equation}\label{dwefregerfgrfgdffwrtretrettyrgfhfghfghsf}
X=Z_{i+1} + b_{i+1}, Y = V_i, Q = H_{i+1}, S = S_i, \alpha = \alpha_i, \epsilon = \epsilon_i,
\end{equation}
(it is easy to check that such substitutions are valid and also that $|Y| = b_i$).  As a result, the equality  $(\ref{dwdffwrtretrettyrgfhfghfghsf})$ becomes the equality $(\ref{dfsdefdsfdsfdssdsfsdfsdf})$. Recall that $(\ref{dfsdefdsfdsfdssdsfsdfsdf})$ implies $(\ref{dfsdsfsdfsdf})$; thus, by reversing the substitutions in $(\ref{dfsdsfsdfsdf})$ we obtain:
\begin{equation}\label{sdererrreerretertretfsdsdg}
P(|Z_{i+1}+b_{i+1} - b_{i}|\ge \epsilon_ib_i) \le {\alpha_i \over \epsilon_i^2}.
\end{equation}
Note that the left side in $(\ref{sdererrreerretertretfsdsdg})$ is equal to the left side in $(\ref{dwdffwsf})$ and thus it remains to bound the right side. To do that, we recall that $\alpha_i$ is defined in $(\ref{sdfsdsdg})$ such that ${\alpha_i \over \epsilon_i^2} = 0.1(1-u)u^i$.
Thus, $(\ref{sdererrreerretertretfsdsdg})$ gives us $(\ref{dwdffwsf})$ for every $i$ and the lemma follows.
\end{proof}

\fi

Since $Z$ is a fixing sequence for $B$ we can approximate $|\beta| = b_0$ by $b_t + \sum_{i=1}^t Z_i$.
The idea is that it is sufficient to compute the fixing sequence $Z$ and $b_m$.
We will prove that it is possible to do so by fixing the cost as well if the cost function is geometric.

\section{Proving Theorem \ref{sdsdkflksdflksdf}}\label{sdjfkjsdfkjskdjfkjsdf}
\begin{proof}

First, we will assume that the $AHE$ algorithm $\mathcal{A}$ is deterministic.
Construct a random $t \times n$ domain-sampling matrix (Definition \ref{deddsfdsfsdf}) $H$.
The space complexity of maintaining $H$ is polylogarithmic and thus it is $o(w)$.
The data stream $D$ and matrix $H$ defines the sequence of vectors $\{V_i\}_{i=0}^t$ using (\ref{sampleone}) and (\ref{sampletwo}) as follows.
$\beta$ is the vector with entries $f_i^k$. $V_i$ is defined as in (\ref{sampleone}) and (\ref{sampletwo}).

For each vector $V_i$ we use algorithm $\mathcal{A}$ to find all heavy elements. In particular, $\mathcal{A}$ will output all
$v_{i,j}$ such that $v_{i,j} \ge \alpha|V_i|$.
To do that we apply $\mathcal{A}$ on the subset of $D_i \subseteq D$ defined by the matrix $H$ as follows.
\begin{equation}\label{ljlsdnvclkjndskjcdskjsdjksdjk}
D_i = \{p_j \in D : \prod_{l=1}^{i} h_{i,p_j} \neq 0\}.
\end{equation}
That is, $D_i$ is a subset of $D$ with all elements that are not zeroed by $H$, up to the $i$-th row.
It is straightforward to see that the entries of $V_i$ are $k$-th powers of the frequencies in $D_i$.
Thus, $\mathcal{A}$ can be used to find the heavy elements in $V_i$.

Let $b_i = |V_i|$, and let $B=\{b_0,b_1,b_2, \ldots b_t\}$.  Note that $b_0 = |V_0| = F_k$.  Construct the sequence $Z$ as defined in (\ref{sderertertretfsdsdeeefwwrwg}).

By Lemma \ref{fixingsequence}, $Z$ is a fixing sequence for $B$.  By (\ref{dfrerdfwdwsf}), $P(|b_t + \sum_{i=1}^t Z_i - b_0| \ge \epsilon b_0) \le 0.2$, which is a $(1 \pm \epsilon)$ approximation of $b_0 = F_k$.
To compute the approximation we only need $b_t$ and $Z_i$ for all $i$.
Note that since $t = O(\log n)$ then with high probability $F_0(V_t) = O(1)$ and therefore we can compute
$b_t = |V_t|$ precisely using $O(\log n)$ bits.
Note that to compute the value of $Z_i$ it is only necessary to know the value of heavy elements for all $V_i$s.
Therefore by applying $\mathcal{A}$ on $D_i$ it is possible to compute $Z_i$.
We conclude that using a sequence of algorithms $\mathcal{A}$ it is possible to approximate $F_k$.

It remains to bound the cost of the algorithm.\footnote{C is the constant from Theorem \ref{wdfkjnwkejfkjwefkjwef}}
Let\footnote{The analysis will work for any $q>0.5$.}
\begin{equation}\label{dfljklsdjfjkjsdfjksd}
q=0.6, u = \left({1\over 2q}\right)^{(1-2/k)\over 3C}.
\end{equation}
Note that $u<1$.
Fix error parameter $\epsilon$ and compute the parameters $\alpha_i$ for $D_i$ according to $(\ref{sdfsdsdg})$.
By $(\ref{welfjlweflkwflwke})$, the cost of computing the heavy element is:
\begin{equation}\label{kwbnfkjwfkjwefjkew}
O({1\over \alpha_i^C}(F_0(D_i))^{1-2/k}).
\end{equation}
Let us bound ${1\over \alpha_i^C}$:
\begin{equation}\label{kwbnfkjwfkjwefjkew}
\alpha_i = 0.1\epsilon_i^2 u^i(1-u) = ((0.1(1-u)u^i)\epsilon)^2 u^i(1-u) = 0.01(1-u)^3\epsilon^2u^{3i}.
\end{equation}
Here the first equality follows from $(\ref{sdfsdsdg})$, the second equality follows from $(\ref{dfdfwdwffdsf})$ and the third equality
follows from direct computations.
It follows that
\begin{equation}\label{fdvjfdkjfdkjfdjk}
\alpha_i^C = (0.01(1-u)^3\epsilon^2)^C u^{3iC} = (0.01(1-u)^3\epsilon^2)^C {1\over (2q)^{i(1-2/k)}}.
\end{equation}
Thus

\begin{equation}\label{kwbnfkjwfkjwefjkew}
{1\over \alpha_i^C}(F_0(D_i))^{1-2/k} \le C_0((2q)^i F_0(D_i))^{1-2/k},
\end{equation}
where $C_0$ is a constant defined as
\begin{equation}\label{kwbnfkjwfkjwefjkew}
C_0 = (0.01(1-u)^3\epsilon^2)^C.
\end{equation}

Note that Definition \ref{deddsfdsfsdf} of matrix $H$ and $(\ref{ljlsdnvclkjndskjcdskjsdjksdjk})$ imply that
\begin{equation}\label{kwbnfkjwfkjwefjkew}
E(F_0(D_i)) = 2^{-i}F_0(D) \le 2^{-i}n.
\end{equation}
Denote sequence $d_i = 2^iF_0(D_i)$.
Then $E(d_i) = d_0$.
Thus, by $(\ref{kwbnfkjwfkjwefjkew})$ we have
\begin{equation}\label{kwbnfkjwfkjwefjkew}
{1\over \alpha_i^C}(F_0(D_i))^{1-2/k} \le C_0((2q)^i F_0(D_i))^{1-2/k} = C_0((2q)^i 2^{-i}d_i)^{1-2/k} =  C_0 (d_iq^i)^{1-2/k}.
\end{equation}

Thus, we can apply Lemma \ref{geometric} with our choice of the random variables $d_i$ and constant $q$ and using $\gamma = 1-2/k$.
We conclude that the costs form a $\theta$-geometric sequence for some $\theta$ that depends only on $q=0.6$
and thus $\theta$ is an absolute constant. Thus, we can conclude by applying Lemma \ref{sdjkkjfsdssfdhdffdsf} that the total cost of the algorithm is \begin{equation}\label{felkmlfkelkfglkefg}
C_0((2q)^i F_0(D_i))^{1-2/k} = C_0 (1-\theta)^{-1}F_0(D)^{1-2/k} = O({1\over \epsilon^{2C}}F_0(D)^{1-2/k}) =  O({1\over \epsilon^{2C}}n^{1-2/k}).
\end{equation}
The theorem is correct when $\mathcal{A}$ is deterministic.

%

Now consider the case when $\mathcal{A}$ is randomized and repeat the above arguments with the following change.
Algorithm $\mathcal{A}$ will be applied on $D_i$ with a probability of error $\delta_i = {1\over 10* 2^i}$.
By the union bound the probability that any instance of the algorithm errs is at most $0.2$.
Note also that the randomness of $H$ is independent of the randomness of the algorithms $\mathcal{A}$.
Thus, the above arguments for the deterministic case conditioned on the event that all instances of the algorithm give correct answers.

The cost of the $i$-th computations will be increased by a factor of $O(i)$ and will become
$$
(iq^i d_i),
$$
where $q = 0.6$. Note that for $i>40$ we have $i(0.6)^i < (0.7)^i$ and thus
$$
(iq^i d_i) \le 40(0.7)^i d_i.
$$
Thus, we can define $q' = 0.7$ and repeat the arguments for the deterministic case.
Therefore the theorem is correct.
\end{proof}

\subsection{Sketches with Geometric Cost}
We will show that it is possible to construct $Z$ using geometric cost in the streaming model.

Let $D = \{d_0,\dots, d_t\}$ be a sequence of random variables,
where $d_0$ is a fixed number and such that $E(d_i) = d_0$. Let $\gamma$ be a parameter and $q$ be a constant such that $0<q<1$.
Define $C = \{c_0,\dots, c_t\}$ as follows:
\begin{equation}\label{sdererrreerretertretfsdsdfdsfdsfdg}
c_0 = (xd_0)^{\gamma};
\end{equation}
\begin{equation}\label{sdererrreerretertretfsdxcdscsdfdsfdsfdg}
c_i = (d_i q^i)^{\gamma},
\end{equation}
where $x$ is a constant that depends on $q$.
\begin{lemma} \label{geometric}
There exist $x$ and $\theta$ that depend only on $q$ and such that $0<\theta <1$ such that $C$ is $\theta$-geometric.
\end{lemma}
\ifNOTtalk
\begin{proof}
Put
\begin{equation}\label{sdeccxzxczxcrecsdfdsvdsdsfg}
\theta= q^{\gamma\over \gamma +1},
\end{equation}
\begin{equation}\label{sderecsdfdsvdsdsfg}
x= 10(1-\theta)^{-1}, \delta = \theta^{1/\gamma}.
\end{equation}
Then we have
\begin{equation}\label{sderecsdcxdsfdsffdsvdsdsfg}
q = \delta\theta, \ \ \ \  q^i = (\delta\theta)^i,
\end{equation}
\begin{equation}\label{sderecsdcxdsfdsffdsvdsdsfg}
{q^i\over x\delta^i} = 0.1\theta^i(1-\theta).
\end{equation}

Further, direct computations imply
\begin{equation}\label{sdererdsdcscsxcrreerretertretfsdsdfdsfdsfdg}
P(c_i \ge c_0\theta^i) =
\end{equation}
\begin{equation}\label{sderecsdsdsdrrreerretertretfsdsdfdsfdsfdg}
P((d_iq^i)^{\gamma} \ge (\delta^{ i}xd_0)^{\gamma}) =
\end{equation}
\begin{equation}\label{sderecsdsdsdrrreerretertretfsdsdfdsfdsfdfdsvdsdsfg}
P(d_i \ge \left({\delta\over q}\right)^{ i}xd_0) \le {q^i\over x\delta^i}.
\end{equation}
The last inequality follows since $E(d_i) = d_0$ and by Markov inequality.

Using $(\ref{sderecsdsdsdrrreerretertretfsdsdfdsfdsfdfdsvdsdsfg})$ and $(\ref{sderecsdcxdsfdsffdsvdsdsfg})$ we conclude that
for each $i$
$$
P(c_i \ge c_0\theta^i) \le 0.1\theta^i(1-\theta).
$$
Therefore, by Definition \ref{weewrfwefwe} $C$ is a $\theta$-geometric sequence and the lemma is correct.
\end{proof}
\fi

\section{Proving Theorem \ref{ekmrmflkerfflkreflker1}}\label{sdljfksjdfjksdfkjsd}

\subsection{A Single Pass}
We now present a proof that it is possible to reduce the number of passes in our algorithm.
The current version of the algorithm requires three passes.
The first pass is needed to compute $F_1$. Using $F_1$ we compute the value of the sampling probability.
During the second pass we sample the stream
and apply the algorithm for heavy elements that has been described in the first part of the paper.
The last pass is needed to compute the exact frequencies of each heavy element.
After these three passes we will apply the Martingale Sketch algorithm to compute the approximation of $F_k$.
To reduce the number of passes to one we will argue that $(1)$ Martingale Sketching can work with approximations instead of precise values and
$(2)$ the algorithm for finding heavy elements can be modified to work without knowing the value of $F_1$.

\subsubsection{Eliminating the First Pass}

\begin{lemma}
The algorithm can be adapted to work without knowledge of $F_1$.
\end{lemma}
\begin{proof}

In this section we prove that it is possible to find a heavy element without knowing the value of $F_1$ in advance.
The algorithm will be modified as follows.
Initially we assume that $F_1 \le 2n$. When $F_1 \ge 2n$ we begin to sample the stream with sampling rate $p=0.5$.
When the length of the stream is doubled, we also halve the sampling probability $p$.
We run in parallel several instances of the algorithm with different parameters.
We keep the winner of the previous execution until the end of all games. If we assume that $m$ and $n$ are polynomially far, there will be at most $O(\log n)$ winners and we can keep all of them with negligible extra memory.

We now present an inductive argument. Without loss of generality suppose that $f_1^k\ge G_k$. This assumption does not affect the space complexity. However,
the assumption implies that if the algorithm outputs $1$ during any of the executions, then $1$ will be the overall winner.
The H\"{o}lder inequality implies that $f_1n^{1-1/k} \ge 2^{\eta}F_1$ for some integer $\eta$.
We will show by induction on $x = \log (F_1/n)$ that
if $f_1^k \ge G_k$ then there exists an algorithm that in one pass finds the heavy element
and uses $O({F_0^{1-2/k}\over 2^{\eta\mu}})$ bits.

First, consider the  base case when $x \le 2$. In this case $F_1 = O(n)$ and the correctness follows from the previous section.

Now, suppose that the statement is correct for $x$ and let us prove it for $x+1$.
Let $\hat{D}$ be the prefix of the stream of length $2^{x-k}n$ and $\check{D}$ be the remaining suffix of the stream.
In the same way, we will use $\hat{x}$ and $\check{x}$ for other variables.

In order to prove the inductive step, we address three cases. Let z be an absolute constant.

\begin{itemize}
\item If $\check{f}_1 \ge {z}f_1$, meaning we have a large amount of the weight of $f_1$ in the suffix, then $\check{f}_1$ is still a heavy element with respect to $F_1$. Note that the right sampling probability for $\check{D}$ ensures that the observations from Section \ref{asdfadsfasfasf} will be correct. Thus, in this case $1$ will be outputted with a constant probability.

\item In the second case, we have $\hat f_1 > (1-z)f_1$ and $\hat G_k < (1-z)^k G_k$. This means we have a large amount of the weight of $f_1$ in the prefix, but not a large amount of the weight of non 1 elements. Thus, we can show that:

$$\hat G_k < (1-z)^k G_k < (1-z)^kf_1^k \le \hat f_1^k$$ And therefore $\hat  f_1^k \ge \hat G_k$.

In this case, as $f_1^k$ is always greater than $G_k$, our inductive assumption holds true, and our algorithm will work with the specified space bound.

\item

 In the last case we have
$\hat f_1 > (1-z) f_1$ and $\hat G_k > (1-z)^k G_k$. This means we have a large amount of 1's in the prefix, but also a large amount of the weight of non one elements.

To solve this problem, we use a hash function to separate the stream $\hat{D}$ into a constant number of substreams.
The Markov inequality implies that $1$ will be a heavy element in the substream it is hashed to with probability $0.99$. Further, we show that our induction holds with the following inequalities. Assume we hash into $y$ substreams, a constant. $C$ is another constant to represent the probabilistic nature of hashing. For the substream that 1 is hashed too, the following holds true:

$$\hat G_k \le {C\over y}G_k < {C\over y}f_1^k \le {C \over y(1-z)^k} \hat f_1^k \le \hat f_1^k $$

and therefore our assumption holds true.

However, we must account for the additional space for creating each substream. By induction, the space complexity decreases by an exponential factor in terms of $2^k$, which accounts for the increase created by the substreams. Thus, summing over all substreams and repeating the experiments to amplify probabilities, the bound still remains sufficiently small. Therefore the statement is true and it is possible to find the heavy element in one pass maintaining the previous space complexity.
\end{itemize}

\end{proof}

\subsubsection{Eliminating the Third Pass}
The key idea is that it is possible to choose an exponentially decreasing error such that the cost of approximation
will form a geometric cost. This can be done by repeating the arguments from Section \ref{sdfjnskdjfkjsdfkjs}.
As a result, we can maintain in one pass all approximations of the heavy elements without increasing the total cost of our algorithm.

Also, the total error that will be introduced will be bounded by $\epsilon F_k$ with a constant probability.
This can be shown by bounding another geometric series of expected weights of all heavy elements in the martingale sequence.

Finally, our AHE algorithm may output non-heavy elements with approximations of the form $\tilde{f}_l\le f_l$.
To guarantee that these additional outputs will not affect the final result, we will do the following.
If we need to approximate all $\rho$-heavy elements we will use our AHE algorithm to find all $\rho^2$-heavy elements and only keep the $1\over \rho$ elements with the largest counters. By doing so we will ensure that all $\rho$-heavy elements will be included with $(1\pm \epsilon)$-approximation.
Also, we will output at most $1\over \rho$ elements that are not $\rho^2$-heavy. Thus in total the weight of these ``noisy'' elements will be at most $\rho F_k$. As a result, we can reduce the error of the noise to be negligible.

\subsection{$k> 3$}
The restriction $k \ge 7$ follows from Lemma \ref{ewrrkfkjlkjwefelkjkweflkjkwef} and Corollary \ref{ewrferfwefwefwerfwerfwerfewf} (bound in $(\ref{sdkjfksdfkjsd})$).
Let $\sigma$ be a sufficiently large constant. In the latter case replacing $2$ and $3$ with $\sigma$ and $\sigma+1$ in the logarithm and choosing sufficiently small $\mu$ decreases the bound to $k>3$ in $(\ref{sdkjfksdfkjsd})$ while increasing the cost of the solution by a constant factor.
See also Observation \ref{sdfkjskjdkjsdfjksdwelklewlkwe}.
In the former case the bound on the number of thick rows in Lemma \ref{ewlkfmwlekfmlwekf} increases the bound on the number of dense rows in Lemma \ref{wdfsdfsdfsdf} by a factor of $\upsilon^2$.
Repeating the arguments of Corollary \ref{ewrferfwefwefwerfwerfwerfewf} in Lemma \ref{ewrrkfkjlkjwefelkjkweflkjkwef} translates into increasing the bound factor on the number of bad rows by a factor of  $4^{\gamma}$. Thus, in
$(\ref{sdkjfksdfkjsd})$ the bound on $k$ is increased by $2$.
To eliminate this problem, observe the following. The new bounds come from the fact that more elements can compete with $1$: namely elements from other rows that have the same signature as the signature of $1$. Thus, the bounds in Corollary \ref{ewrferfwefwefwerfwerfwerfewf} can be decreased by a factor of $2^{\rho\gamma}$, and the bound in $(\ref{sdkjfksdfkjsd})$ still holds.

%
%
%
%

\bibliographystyle{plain}
\bibliography{Bibliography}

{\appendix}

%
%
%

\section{Proving Lemma \ref{sdfsdfsdfsdfnsdvnsdvmnsdmnv}}\label{bla}
In this section we prove Lemma \ref{sdfsdfsdfsdfnsdvnsdvmnsdmnv}.
The proof of Lemma \ref{sdfsdfsdfsdfnsdvnsdvmnsdmnv} has two main steps. First, we show (in Section \ref{sdfgsdgsdgsdg} ) that there exists a pair with large value of $|S_\beta(\alpha)|$. Second, we show (in Section \ref{dfvsdrgdrgrgrg}) that the number of bad rows is small for any pair $\alpha, \beta$. As a result, the number of great rows for the chosen pair must be sufficiently large.

\subsubsection{Bounding $|S_\beta(\alpha)|$}\label{sdfgsdgsdgsdg}


In this section we show that there exists at least one pair $(\alpha, \beta)$ with a sufficiently large $|S_{\beta}(\alpha)|$ that satisfies Lemma \ref{sdfljsdldjfljeflkjfd}.

\begin{lemma}\label{cdvnjlcvljcx}
Let $N$ be a parameter and let $V\in \mathbb{R}^N$ be a vector with integer entries $v_i$ where $0\le v_i \le m$.
For $u=1,\dots, \lceil \log{m}\rceil$ define $Q_u = \{i : 2^{u-1} \le v_i < 2^{u}\}$.
Let $C$ be an absolute constant and let $\{\rho_u\}_{u=1}^{\infty}$ be a sequence of non-negative real numbers such
that
\begin{equation}\label{sdfljedfjwlseeee}
\sum_{u=1}^{\infty} \rho_u \le C.
\end{equation}
Then there exists at least one integer $u \in\{1,\dots, \lceil \log{m}\rceil\}$ such that
\begin{equation}\label{sdfljedfjwls}
\sum\limits_{i\in Q_u} v_i \ge {\rho_u|V| \over C}
\end{equation}
\end{lemma}

\ifNOTtalk
\begin{proof}
Assume, that no such $u$ exists. Then
\begin{equation}
|V| = \sum\limits_{u=1}^{\lceil \log{m}\rceil}\sum_{i\in Q_u}v_i < \frac{|V|}{C} \sum\limits_{u=1}^{\lceil \log{m}\rceil}\rho_u \le |V|.
\end{equation}
Here the equality follows from the fact that $v_i$ are non-negative integers, the first inequality follows from our assumption and the last inequality follows from
$(\ref{sdfljedfjwlseeee})$.
Thus, $|V|<|V|$ and our assumption must be wrong.
\end{proof}
\fi



\begin{corollary}\label{wslrkjkfljsjdfl}
There exists\footnote{It is important to note that the lemma does not hold in general. Indeed, $f_1\over 2^{u+1}u^2$ might be larger than $r_{-0.5\eta}$ for small values of $u$.  However, for our algorithm it is sufficient to establish the existence of at least one such $u$.} $u> 1.5\eta+\Psi$ such that
\begin{equation}\label{edfkjksjdfkjsdfkje12}
|S_u(-0.5\eta)| \ge {f_1\over 2^{u+2}u^2}.
\end{equation}
\end{corollary}
\ifNOTtalk
\begin{proof}
Denote:
\begin{equation}\label{sdnjsdfkjlsdfkj}
\kappa = -0.5\eta.
\end{equation}
Consider $M_{\kappa}$ and let $IND$ be a set of indices of rows in $M_\kappa$ defined as $IND = \{i: f_1(\kappa, i) \le 2^{1.5\eta+\Psi}\}$. We have:
\begin{equation}\label{edfkjksjdfkjsdfkj}
\sum_{i\in IND} f_1(\kappa, i) \le 2^{1.5\eta+\Psi}|IND| \le 2^{1.5\eta+\Psi}r_{\kappa}.
\end{equation}
Here the first inequality follows from the definition of $IND$ and the second inequality follows since $IND$ is a set or rows in $M_\kappa$ and thus its cardinality cannot exceed the total number of rows $r_{\kappa}$.
Note that
\begin{equation}\label{sdfngkdsfkjsdffkjfdkjdfkjdf}
r_{\kappa} = {F_1\over t_\kappa} \le {C_2n2^{\kappa}\over n^{1-1/k}}= {C_22^{\kappa}n^{1/k}}.
\end{equation}
Here the first equality and the first inequality follow from the definitions and assumptions from Table \ref{table1} and Table \ref{table2}.
Thus,
\begin{equation}\label{dsfjndskjfkjsdfkjds}
\sum_{i\in IND} f_1(\kappa, i) \le 2^{1.5\eta+\Psi} (C_22^{\kappa}n^{1/k}) = C_22^{\eta+\Psi}n^{1/k}\le {f_1\over 2^{\Psi}}.
\end{equation}
Here the first inequality follows from $(\ref{sdfngkdsfkjsdffkjfdkjdfkjdf})$ and $(\ref{edfkjksjdfkjsdfkj})$, the equality follows from $(\ref{sdnjsdfkjlsdfkj})$  and the last inequality follows from the assumption on $f_1$ from Table \ref{table1}. Since $\Psi>1$ we have:
\begin{equation}\label{dsjnhfksdjkjsd}
\sum_{i\in IND} f_1(\kappa, i) \le 0.5f_1.
\end{equation}

\noindent
Consider vector $V \in R^{r_\kappa}$ with entries $\{v_i\}_{i=1}^{r_\kappa}$ that are defined as follows:
\begin{equation}\label{slnflnesfljld}
v_i = f_1(\kappa, i) \text{ if } i\notin IND \text{ and } v_i = 0 \text{ otherwise.}
\end{equation}

\noindent
The following inequality follows from the definition of $V$ and $(\ref{dsjnhfksdjkjsd})$:

\begin{equation}\label{dfdssdfsdfdsf}
|V| \ge 0.5f_1.
\end{equation}

\noindent
Consider $\rho_u = 1/u^2$.
By applying Lemma \ref{cdvnjlcvljcx} we conclude\footnote{Recall that $\sum_{u=1}^{\infty} \rho_u = {\pi^2\over 6} <2$.} that there exists $u$ such that

\begin{equation}\label{dsjnhfksdjkjsd2}
\sum\limits_{i\in Q_u} v_i \ge {|V| \over 2u^2},
\end{equation}

\noindent
where $Q_u$ defined in Lemma \ref{cdvnjlcvljcx}.
Note that $(\ref{dfdssdfsdfdsf})$ implies that $|V|$ is strictly positive and thus if $u$ satisfies $(\ref{dsjnhfksdjkjsd2})$ then $u > 1.5\eta+\Psi$. Indeed, if $u\le 1.5\eta+\Psi$ then it follows that $Q_u \subseteq IND$ and thus $\sum\limits_{i\in Q_u} v_i =0$.
It is easy to check that if $u$ satisfies $(\ref{dsjnhfksdjkjsd2})$ then
\begin{equation}\label{ereregfdfgdfgdfgdf}
Q_u = S_u(\kappa).
\end{equation}

Summarizing all of the above, we conclude that there exists at least one integer $u > 1.5\eta+\Psi$ such that
\begin{equation}\label{sdfsdfsdfsdf}
|S_u(\kappa)|2^u > \sum_{i\in S_u(\kappa)} f_1(\kappa, i) = \sum\limits_{i\in Q_u} v_i \ge {|V|\over 2u^2} \ge {f_1\over 4u^2}.
\end{equation}
Here the first inequality follows from the definition of $S_u(\kappa)$, the first equality follows from $(\ref{ereregfdfgdfgdfgdf})$ and $(\ref{slnflnesfljld})$, the second inequality follows from $(\ref{dsjnhfksdjkjsd2})$
and the last inequality follows from $(\ref{dfdssdfsdfdsf})$.
Scaling $(\ref{sdfsdfsdfsdf})$ by a factor of $2^{-u}$ gives $(\ref{edfkjksjdfkjsdfkje12})$.
\end{proof}
\fi

\begin{corollary}\label{rfmkjrsdkdfskjdfgkdgf}
Consider $u$ as in Corollary \ref{wslrkjkfljsjdfl} and suppose that $\beta$ is defined as in Algorithm \ref{sdjnfsdfsjdfskjd}. Then $\beta > 0.8\Psi$.
\end{corollary}
\begin{proof}
By the description of Algorithm \ref{sdjnfsdfsjdfskjd} we have that $\beta \ge 0.8u$. Corollary \ref{wslrkjkfljsjdfl} implies that $u>\Psi$.
\end{proof}

\includegraphics[scale=0.5]{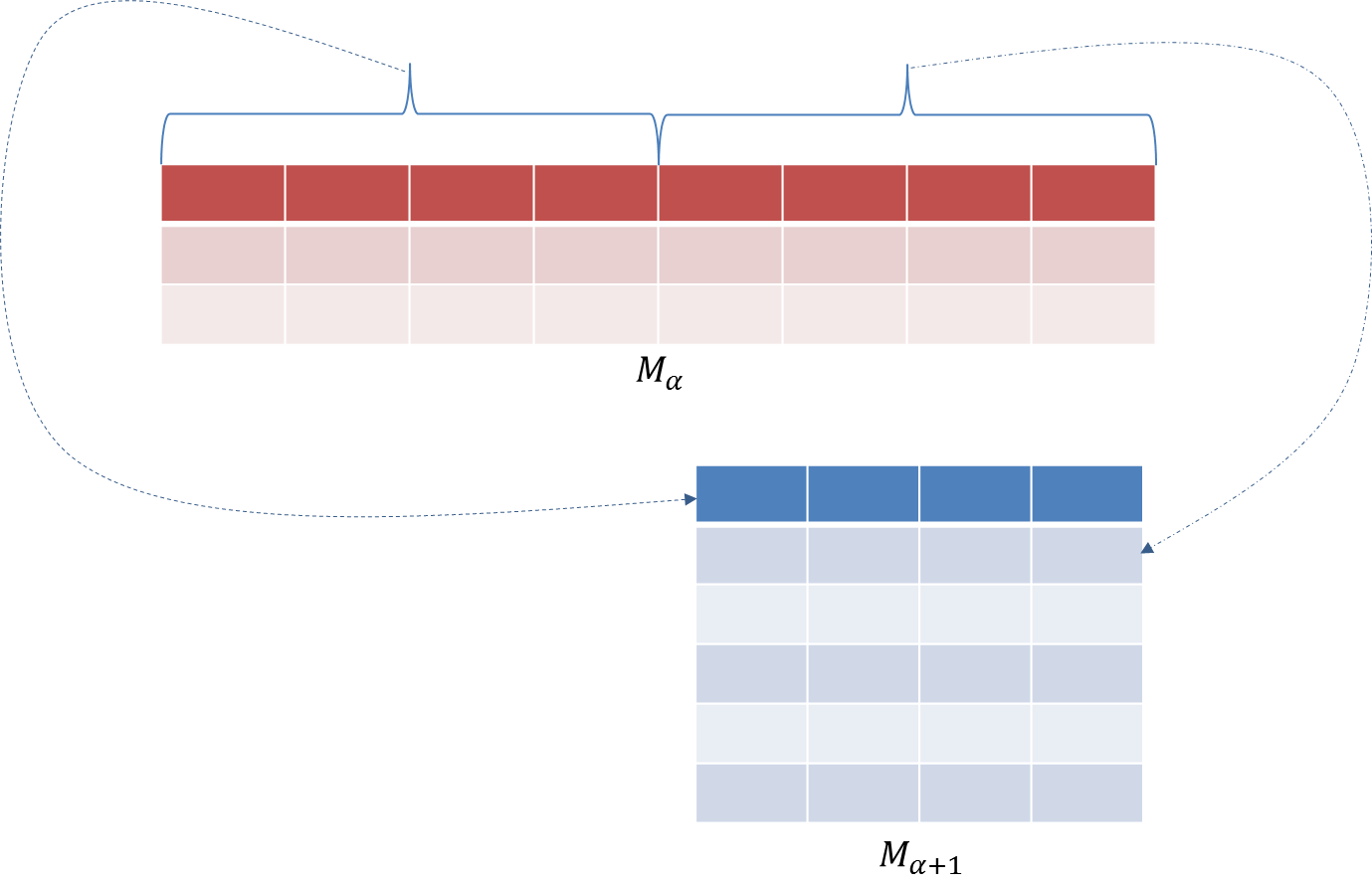}
$$\text{Picture 3: Relation between matrices}$$

\begin{corollary}\label{dflkmgvldsfgldsflksdflk}
Consider $u$ as in Corollary \ref{wslrkjkfljsjdfl}.
Given $u$, Algorithm \ref{sdjnfsdfsjdfskjd} defines a sequence of pairs $\alpha, \beta$.
There exists at least one pair $(\alpha, \beta)$ from that sequence such that
\begin{equation}\label{wefjnnewlkjkjwkerdkjwe}
|S_{\beta}(\alpha)| \ge {f_1\over 2^{\beta+0.5\mu\beta}}.
\end{equation}
\end{corollary}
\ifNOTtalk
\begin{proof}

Consider $u$ as in Corollary \ref{wslrkjkfljsjdfl}.
Algorithm \ref{sdjnfsdfsjdfskjd} distinguishes between two cases: when $u\le 20\eta$ and when $u> 20\eta$.

Consider the case when $u\le 20\eta$. In this case Algorithm \ref{sdjnfsdfsjdfskjd} defines only one pair: $\alpha = -0.5\eta$ and $\beta=u$. For this pair Corollary \ref{wslrkjkfljsjdfl} immediately implies
\begin{equation}\label{dskjckjsdcjkvcsksck}
|S_{\beta}(\alpha)| \ge {f_1\over 2^{\beta+2}\beta^2} > {f_1\over 2^{\beta+10}\beta^3}.
\end{equation}

Consider the case
\begin{equation}\label{wekmfrlmkwdfkljsdkjfwdksjfsdf}
u>20\eta.
\end{equation}
In this case Algorithm \ref{sdjnfsdfsjdfskjd} defines a sequence of pairs, where $\alpha = u/5$ and
$\beta \in \{0.8u-0.5\eta-2, \dots, u\}$.
We will show that for at least one pair:
\begin{equation}\label{wekmfrlmkwdfkljsdkjfwdksjfsdf}
|S_{\beta}(\alpha)| \ge {f_1\over 2^{\beta+10}\beta^3}.
\end{equation}

If we assume that $(\ref{wekmfrlmkwdfkljsdkjfwdksjfsdf})$ is correct then the proof of Corollary \ref{dflkmgvldsfgldsflksdflk} follows immediately.
Indeed, Corollary \ref{rfmkjrsdkdfskjdfgkdgf} and Fact \ref{sdfjsdkjfkjsdf} imply that $2^{10}\beta^3 \le 2^{0.5\mu\beta}$ and $(\ref{wefjnnewlkjkjwkerdkjwe})$ follows.

It remains to show that $(\ref{wekmfrlmkwdfkljsdkjfwdksjfsdf})$ is correct.
Consider a fixed row from $M_{-0.5\eta}$ with index $i'$ such that
\begin{equation}\label{kwerketfetrfkuekfkfjfewd}
i'\in S_u(-0.5\eta).
\end{equation}
By the definition of the matrices $M_{\chi}$ that is given in Section \ref{F} the subset of stream $D$ that corresponds to a one row in $M_{-0.5\eta}$ is equal to a subset of $D$ that correspond to $2^{\alpha+0.5\eta}$ rows in $M_{\alpha}$. An example is presented in Picture 3.
Let $IND(i')$ be the set $IND(i') \subseteq [r_\alpha]$ of row indices in $M_\alpha$ that corresponds to the row index $i'$ in $M_{-0.5\eta}$.
The following facts can be easily verified.
\begin{fact}\label{dskkdfksdfkjfskjds}
$\sum_{i\in IND(i')}f_1(\alpha,i) = f_1(-0.5\eta,i').$
\end{fact}
\begin{fact}\label{sdfjnskdjfkjsdf}
If $i'\neq i''$ then sets $IND(i')$ and $IND(i'')$ are disjoint.
\end{fact}

It follows that
\begin{equation}\label{fsdljvlkjskdkjsdkjdskjsd}
|IND(i')| = 2^{\alpha+0.5\eta}
\end{equation}
Consider the following partition of $IND$ into two subsets:
\begin{equation}\label{wekmfrlmkwdfkljsdkjfwdksjfsdf1}
IND(i')_\le = \{i: i\in IND(i'), f_1(\alpha,i) \le {2^{u-\alpha-0.5\eta-2}}\},
\end{equation}
\begin{equation}\label{wekmfrlmkwdfkljsdkjfwdksjfsdf2}
IND(i')_> = \{i: i\in IND(i'), f_1(\alpha,i) > {2^{u-\alpha-0.5\eta-2}}\}
\end{equation}
We have
\begin{equation}\label{wekmfrlmkwdfkljsdkjfwdksjfsdf3}
\sum_{i\in IND(i')_\le} f_1(\alpha,i) \le |IND(i')|{2^{u-\alpha-0.5\eta-2}} = 2^{\alpha+0.5\eta}{2^{u-\alpha-0.5\eta-2}} = 2^{u-2} \le 0.5f_1(-0.5\eta,i').
\end{equation}
Here the first inequality follows from $(\ref{wekmfrlmkwdfkljsdkjfwdksjfsdf1})$, the first equality follows from $(\ref{fsdljvlkjskdkjsdkjdskjsd})$ and the last inequality follows from $(\ref{kwerketfetrfkuekfkfjfewd})$.
Also,
\begin{equation}\label{wekmfrlmkwdfkljsdkjfwdksjfsdf4}
\sum_{i\in IND(i')_\le} f_1(\alpha,i) + \sum_{i\in IND(i')_>} f_1(\alpha,i) = \sum_{i\in IND(i')} f_1(\alpha,i) = f_1(-0.5\eta,i').
\end{equation}
Here the first equality follows from $(\ref{wekmfrlmkwdfkljsdkjfwdksjfsdf1})$ and $(\ref{wekmfrlmkwdfkljsdkjfwdksjfsdf2})$
and the second equality follows from Fact \ref{dskkdfksdfkjfskjds}.
Together, $(\ref{wekmfrlmkwdfkljsdkjfwdksjfsdf4})$ and $(\ref{wekmfrlmkwdfkljsdkjfwdksjfsdf3})$ imply:
\begin{equation}\label{ijsdfodsfgkjhgksjdgfkdskjdk}
\sum_{i\in IND(i')_>} f_1(\alpha,i) > 0.5f_1(-0.5\eta,i').
\end{equation}
Let
\begin{equation}\label{ljrenfkjsdfkjksdfkjlsd}
SET = \cup_{i'\in S_u(-0.5\eta)}IND(i')_>.
\end{equation}
Summing over all $i'$ we obtain
\begin{equation}\label{wejirowertuofregudsuiogfudsgfuoi1}
\sum_{i\in SET} f_1(\alpha,i) = \sum_{i'\in S_u(-0.5\eta)}\sum_{i\in IND(i')_>} f_1(\alpha,i) >
\end{equation}
\begin{equation}\label{wejirowertuofregudsuiogfudsgfuoi2}
0.5\sum_{i'\in S_u(-0.5\eta)}f_1(-0.5\eta,i') \ge 0.5  |S_u(-0.5\eta)|2^{u-1} \ge 0.5{f_1\over 2^{u+2}u^2}2^{u-1} = {f_1\over 16u^2}.
\end{equation}
Here the first equality follows from Fact \ref{sdfjnskdjfkjsdf}, the first inequality follows from $(\ref{ijsdfodsfgkjhgksjdgfkdskjdk})$, the second inequality follows from the definition of $S_u(-0.5\eta)$ and the third inequality follows from our choice of $u$ and from $(\ref{edfkjksjdfkjsdfkje12})$.
To conclude, we obtain
\begin{equation}\label{wejirowertuofregudsuiogfudsgfuoi}
\sum_{i\in SET} f_1(\alpha,i) > {f_1\over 16u^2}.
\end{equation}

Denote $\omega=0.8u-0.5\eta-2$.
Next let us show the following fact.
\begin{fact}\label{sdfljskdjfkjsdf}
\begin{equation}\label{fsdlnvjdfnjdsjfjsdjj}
SET \subseteq \cup_{\beta = (\omega)}^uS_{\beta}(\alpha).
\end{equation}
\end{fact}
\begin{proof}
Let $i\in SET$. In particular, $i\in [r_\alpha]$ and thus there must exist some $\beta$ such that
\begin{equation}\label{elrkfldjfldjsfsd}
i\in S_\beta(\alpha).
\end{equation}
Let us bound the value of $\beta$.
First, observe that since $i\in SET$ it follows from $(\ref{ljrenfkjsdfkjksdfkjlsd})$ and $(\ref{wekmfrlmkwdfkljsdkjfwdksjfsdf2})$ that \begin{equation}\label{dsfljnkjdsfkjdffkjdsfk}
f_1(\alpha,i) > {2^{u-\alpha-0.5\eta-2}}.
\end{equation}
Therefore we obtain a lower bound on $\beta$:
\begin{equation}\label{ldsljsdfkjlsdfkjlsdfkj}
\beta > u-\alpha-0.5\eta-2.
\end{equation}
Indeed, if $\beta \le u-\alpha-0.5\eta-2$, we obtain a contradiction with $(\ref{dsfljnkjdsfkjdffkjdsfk})$:
\begin{equation}\label{wejirowertuofregudsuiogfudsgfuoi1}
f_1(\alpha,i) < 2^{\beta}\le {2^{u-\alpha-0.5\eta-2}}.
\end{equation}
Here the first inequality follows from $(\ref{elrkfldjfldjsfsd})$.
To obtain the upper bound on $\beta$ observe that
\begin{equation}\label{sdlfklkjdflkjfkjlfkljdfdw}
f_1(\alpha,i) \le f_1(-0.5\eta,i') < 2^u.
\end{equation}
Here the first inequality follows from Fact \ref{dskkdfksdfkjfskjds} and the second inequality follows from $(\ref{kwerketfetrfkuekfkfjfewd})$.
Thus, we obtain the upper bound on $\beta$:
\begin{equation}\label{jldskjfsdkjfsdkjfkjsdf}
\beta \le u.
\end{equation}
Indeed, if $\beta>u$ then we obtain a contradiction\footnote{Recall that $\beta$ and $u$ are integers.} with $(\ref{sdlfklkjdflkjfkjlfkljdfdw})$:
\begin{equation}\label{sdlfklkjdflkjfkjlfkljdfdw1}
f_1(\alpha,i) \ge 2^{\beta-1} \ge 2^{u}.
\end{equation}

Recall that we are proving the lemma for the case $(\ref{wekmfrlmkwdfkljsdkjfwdksjfsdf})$.
In that case Algorithm \ref{sdjnfsdfsjdfskjd} and Corollary \ref{wslrkjkfljsjdfl} imply that $\alpha = u/5, 20\eta < u$ and also $u>\Psi$. Direct computations imply that
\begin{equation}\label{fdljvdsfkjkjdskjdskjdf}
u-\alpha-0.5\eta -2\ge 0.25u.
\end{equation}
Our bounds on $\beta$ in $(\ref{jldskjfsdkjfsdkjfkjsdf})$ and $(\ref{ldsljsdfkjlsdfkjlsdfkj})$ together with $(\ref{fdljvdsfkjkjdskjdskjdf})$ completes the proof.
\end{proof}

Let us finish the proof that $(\ref{wekmfrlmkwdfkljsdkjfwdksjfsdf})$ is correct.
\begin{equation}\label{reljfnkjldfgkjndsfgkdsfkjnsndjf}
\sum_{\beta= (\omega)}^u |S_\beta(\alpha)| 2^{\beta} >
\sum_{\beta= (\omega)}^u \sum_{i\in S_\beta(\alpha)} f_1(\alpha,i) \ge \sum_{i\in SET} f_1(\alpha,i) \ge {f_1\over 16u^2}.
\end{equation}
Here the first inequality follows from the definition of $S_\beta(\alpha)$, the second inequality follows from Fact \ref{sdfljskdjfkjsdf} and the last inequality follows from $(\ref{wejirowertuofregudsuiogfudsgfuoi})$.
Finally, assume that
\begin{equation}\label{dsjfkdsjfkjsdkjsfd}
\forall \beta\in \{(\omega),\dots, u\} : |S_{\beta}(\alpha)| < {f_1\over 2^{\beta+10}\beta^3}.
\end{equation}
Then
\begin{equation}\label{ldsfkjsdfkjsdfkjsdfkjsd}
{f_1\over 16u^2} >
\sum_{\beta= (\omega)}^u {f_1\over 2^{\beta+10}\beta^3}2^{\beta} > \sum_{\beta= (\omega)}^u |S_\beta(\alpha)| 2^{\beta} > {f_1\over 16u^2}.
\end{equation}
Here the first inequality follows from direct computations, the second inequality follows from the assumption $(\ref{dsjfkdsjfkjsdkjsfd})$ and the last inequality follows from $(\ref{reljfnkjldfgkjndsfgkdsfkjnsndjf})$.
This is a contradiction and therefor the assumption $(\ref{dsjfkdsjfkjsdkjsfd})$ in wrong.
Thus, we have shown the correctness of $(\ref{wekmfrlmkwdfkljsdkjfwdksjfsdf})$.

\noindent

\end{proof}
\fi

\begin{fact}\label{sdfjsdkjfkjsdf}
If $\beta>0.8\Psi$ then $2^{10}\beta^3 \le 2^{0.5\mu\beta}$.
\end{fact}
\begin{proof}
The following inequality is well known. For $x>0$:
\begin{equation}\label{dsjfkjsdkjsdkjsd}
e^x \ge {x^4\over 4!}
\end{equation}
In particular, let $a,b>1$. Then for $x>(4!)a^4b$:
\begin{equation}\label{fdkfdkgdkfjdkjf}
e^{x/a} \ge {x^4\over 4!a^4} \ge bx^3.
\end{equation}
Thus, putting $a={2(\log_2e)\over \mu}$ and $b=2^{10}$ we obtain the following.
If $\beta > (4!)a^4b$ then
\begin{equation}\label{fdkfdkgdkfjdkjf}
2^{0.5\mu\beta} \ge 2^{10}\beta^3.
\end{equation}
It remains to show that $\beta > (4!)a^4b$. This is indeed true since $\beta>0.8\Psi$ and by the definition of $\Psi$ from Table \ref{table1}.
\end{proof}

\subsubsection{Bounding the Number of Bad Rows}\label{dfvsdrgdrgrgrg}


\begin{lemma}\label{wdfsdfsdfsdf}
The number of $(\lambda, \phi, \tau)$-dense rows\footnote{Recall that the definitions of the bad and great rows are given in Definitions \ref{dsfsefsdfsdfdsfdsfdsfdsf} and \ref{ewfwefwefwef}.} is at most
${G_k \over \lambda^{k-1}t_{\alpha} \phi {\tau}}\text{.}$
\end{lemma}
\ifNOTtalk
\begin{proof}
The lemma follows from Definition \ref{dsfsefsdfsdfdsfdsfdsfdsf}
and the following bound on the total weight of all elements in $T_{\lambda}$:
$$\sum_{l\in T_{\lambda}} f_l = \sum_{l\in T_{\lambda}} f_l^k {1\over f_{l}^{k-1}} \le  {1\over \lambda^{k-1}}\sum_{l\in T_{\lambda}} f_l^k \le  {1\over \lambda^{k-1}} G_k.$$
\end{proof}
\fi

\begin{corollary}\label{ewrferfwefwefwerfwerfwerfewf}
Let $\gamma\ge \beta$ be two parameters and let $k \ge 5$.
The number of $(\gamma,\beta)$-bad rows is at most\footnote{Corollary \ref{ewrferfwefwefwerfwerfwerfewf} is not informative if ${f_1 \over 2^{\beta+\mu\gamma}}\ge r_\alpha$. However, we will apply Corollary \ref{ewrferfwefwefwerfwerfwerfewf} in a context of Corollary \ref{dflkmgvldsfgldsflksdflk} when the aforementioned trivial case does not happen.}
$f_1 \over 2^{\beta+\mu\gamma}$.
\end{corollary}
\ifNOTtalk
\begin{proof}

Recall that $\gamma \ge \beta$ and $u\ge 1.5\eta+\Psi$.
Apply Lemma \ref{wdfsdfsdfsdf} with these parameters:
$$
\phi = \xi, \lambda = TR, \tau = IC.
$$
Denote
$$
X = {G_k \over \lambda^{k-1}t_{\alpha} \phi {\tau}}
$$
Then, by substitution and the assumptions on the value of $G_k$ (see Table \ref{table1}):
$$
X \le {2^{k\eta}n\over 2^{(\gamma - \alpha+\eta-1)(k-1)} {n^{1-1/k}2^{-\alpha}}\xi2^{\beta-7}}.
$$
Denote
$$
Y = 2^{\eta+2\Psi-\beta-\mu\gamma}n^{1/k}.
$$
Then by assumptions on the value of $f_1$ (see Table \ref{table1}):
$$
Y \le {f_1\over 2^{\beta+\mu\gamma}}.
$$
By Lemma \ref{wdfsdfsdfsdf} it is sufficient to show that $X\le Y$.
That is equivalent to showing that
$$
k\eta - (\gamma - \alpha+\eta-1)(k-1) + \alpha + (\log_2 3)\gamma + \mu\gamma + 7\le \eta+2\Psi-\beta-\mu\gamma.
$$
After some work we obtain an equivalent statement:
$$
k\alpha+\beta+k+6 \le \gamma(k-1-2\mu-\log_2 3) +2\Psi
$$
Simplifying further it is sufficient to show that (since $\Psi$ is sufficiently large):
\begin{equation}\label{sdkjfksdfkjsd}
k\alpha \le \gamma(k-2-2\mu-\log_2 3)
\end{equation}
If $\alpha \le 0$ then the statement is true for $k\ge 4$.
If $\alpha>0$ then $\alpha \le {\gamma\over 4}$ because $\gamma >\beta$ and $\beta \ge u -\alpha$ and $\alpha \ge {u \over 5}$. Therefore, the statement is true\footnote{It is possible to obtain better bounds by choosing smaller $\alpha$. We defer the analysis to the future versions.} for $k\ge 5$.
\begin{observation}\label{sdfkjskjdkjsdfjksdwelklewlkwe}
Let $\sigma$ be sufficiently large constant.
If we replace $2^\gamma$ and $3^\gamma$ in our game with $\sigma^\gamma$ and $(\sigma+1)^\gamma$ then the
analysis will be still correct (with larger constants).
Thus, the bound $(\ref{sdkjfksdfkjsd})$ will work for any constant $k>3$.
\end{observation}

\end{proof}
\fi

\begin{corollary}\label{dfjsdkjfkljsdfksjkdf}
Let $\beta$ be a parameter and let $k\ge 5$.
The number of $\beta$-bad rows is at most $f_1 \over 2^{\beta+0.9\mu\beta}$.
\end{corollary}
\begin{proof}
Summing over all $\gamma \ge \beta$ and using Corollary \ref{ewrferfwefwefwerfwerfwerfewf} we obtain that the number of $\beta$-bad rows is at most
$$
\sum_{\gamma\ge \beta}{f_1 \over 2^{\beta+\mu\gamma}} \le {f_1 \over 2^{\beta+\mu\beta}}{1\over 1-2^{-\mu}} \le {f_1 \over 2^{\beta+0.9\mu\beta}}.
$$
Here the first inequality follows from direct computations and the second follows from Fact \ref{rfmkjrsdkdfskjdfgkdgf} and the definition of $\Psi$.
\end{proof}

\subsubsection{Winning the $(\alpha,\beta)$-game}\label{dfjksdfkjsdfkjsdfkjsfd}

In this section we summarize the results of the two previous sections and prove the existence of a pair with a lower bound on the number of great rows that are from $S_{\beta}(\alpha)$.

\begin{lemma}\label{erregregegergregregrft}
Let $k\ge 5$. There exists a pair $\alpha,\beta$ such that the number of great rows that are from $S_{\beta}(\alpha)$ is at least
$$
{f_1 \over 2^{\beta+0.5\mu\beta-1}}.
$$
\end{lemma}
\ifNOTtalk
\begin{proof}
Consider the pair $\alpha, \beta$ from Lemma \ref{dflkmgvldsfgldsflksdflk}. For this pair we have:
$$
0.5|S_{\beta}(\alpha)| \ge {f_1 \over 2^{\beta+0.5\mu\beta-1}}\ge {f_1 \over 2^{\beta+0.9\mu\beta}}
$$
Here the first inequality follows from equation $(\ref{wefjnnewlkjkjwkerdkjwe})$ and the second inequality follows from direct computations, Fact \ref{rfmkjrsdkdfskjdfgkdgf}, and the definition of $\Psi$.
By Corollary \ref{dfjsdkjfkljsdfksjkdf} this gives us the upper bound on the number of all bad rows.
Therefore at most half of all rows in $S_{\beta}(\alpha)$ are bad and at least half of the rows are great which proves the lemma.
\end{proof}
\fi

Now we can prove the main lemma of this section. Recall the statement of the lemma.
\begin{lemma}(Lemma \ref{sdfsdfsdfsdfnsdvnsdvmnsdmnv})
\existenceLemma
\end{lemma}
\begin{proof}
Corollary \ref{erregregegergregregrft} implies the existence of a pair $\alpha, \beta$
such that:
$$
X \ge {f_1 \over 2^{\beta+0.5\mu\beta-1}}.
$$

Therefore,
$$
\frac{wX2^\beta }{t_{\alpha}} \ge {f_1 \over 2^{\beta+0.5\mu\beta-1}} {w2^{\beta} \over t_{\alpha}}.
$$
Substituting the definition from Table \ref{table2} we obtain :
$$
\frac{wX2^\beta }{t_{\alpha}} \ge {2^{\eta+2\Psi} n^{1/k}\over 2^{\beta+0.5\mu\beta-1}} {n^{1-2/k}2^{\beta+\alpha} \over  n^{1-1/k}2^{\mu\beta}} .
$$
Thus, it is sufficient to show that:
$$
\eta+2\Psi -\beta-0.5\mu\beta+1 +\beta+\alpha -\mu\beta \ge 0.
$$
Indeed, if $\alpha = -0.5\eta \le 0 $ then it is sufficient to show that:
$$
\eta \ge 3\mu\beta.
$$
Otherwise, $\alpha = u/5$ and it is sufficient to show that
$$
\alpha \ge 1.5\mu\beta.
$$
In both cases the bounds follow from the definitions and bounds on $\beta, \eta, \mu$ from Algorithm \ref{sdjnfsdfsjdfskjd} and Table \ref{table1}.
%
\end{proof}
\fi

\newpage
\section{Proving Theorem \ref{sigok}}\label{signaturesproofs}
In this section we prove Theorem \ref{sigok}. We begin with a definition.

\begin{definition}
Let $i$ be a fixed row and $q$ be a fixed player. Assume that an $(\alpha,\beta)$-game is being played and that there is a $\gamma$-th round for the $i$-th row. Let $a \neq b\in [n]$ be fixed. Denote by $\Gamma_{q,i}$ the pool from Section \ref{dskfjkjasdskjsdf} for the $q$-th player.
Denote by $\Upsilon(i,q,\gamma,a,b)$ the random event that both $a,b\in \Gamma_{i,z}$ and that $R_{\varrho\gamma}(a) = R_{\varrho\gamma}(b)$.
\end{definition}

\begin{fact}\label{sdgfdfgdfgdsg} For any fixed $i,q,\gamma,a,b$

$$P(\Upsilon(i,q,\gamma,a,b)) \le ({2^{\beta-\Psi}\over t_{\alpha}})^2{1\over 2^{\varrho\gamma}}.$$
In particular,
$$P(\Upsilon(i,q,\gamma,1,b)|1\in \Gamma_{i,z}) \le {2^{\beta-\Psi}\over t_{\alpha}}{1\over 2^{\varrho\gamma}}.$$
\end{fact}
\ifNOTtalk
\begin{proof}
$$
P(\Upsilon(i,q,\gamma,a,b)) = P(a,b\in \Gamma_{i,z})P(R_{\varrho\gamma}(a) = R_{\varrho\gamma}(b)) =
$$
$$
(P(a\in \Gamma_{i,z}))^2P(R_{\varrho\gamma}(a) = R_{\varrho\gamma}(b)) \le \left({2^{\beta-\Psi}\over t_{\alpha}}\right)^2{1\over 2^{\varrho\gamma}}.
$$
Here the first equality follows since the events $(a,b\in \Gamma_{q,i})$ and $(R_{\varrho\gamma}(a) = R_{\varrho\gamma}(b))$ are independent.
The second equality follows since $g$ is a pairwise independent function. The last inequality follows by definitions of $g$ and the signature.
The second claim of the lemma follows from the pairwise independence of the sampling hash function $g$.
\end{proof}
\fi

In the first and second phases, observe that elements from rows $i,\dots, i+2^\gamma$ can compete with $1$ during the $\gamma$-th round. Thus, the number of ``bad'' rows will increase; we show that it can increase by a factor of at most $4^{\gamma}$.
Otherwise, the proof remains correct and the claims remain for larger values of $k$.
We first redefine a dense row (see Def. \ref{ewfwefwefwef} ) as follows (we use ``thick'' instead)
\begin{definition}\label{ewfewwefew}
Let $T_\lambda  = {  \{  l :  f_l > \lambda, l > 1 \}}$. We say that $i\in [r_\alpha]$ is a $(\lambda, \phi, \tau,\upsilon)$-thick row  if:
\begin{equation}\label{wefewefwefwefwef}
| \{  l: \sum_{a=i}^{i+\upsilon} f_{l}(\alpha,a) > {\tau}, l \in T_{\lambda} \}| > t_{\alpha} \phi.
\end{equation}
Thus, we count over a range of rows instead of a single row.
\end{definition}

In the reminder of this section we use the definitions of ``faulty'' and ``perfect'' rows instead of the definitions of ``bad'' and ``great'' rows.
\begin{definition}\label{dejfnkjdsfkjfdkjdf}
Row $i\in [r_\alpha]$ is $\beta$-faulty if there exists $\gamma \ge \beta$ such that $i$ is
$(TR, \xi, \beta-7, 2^\gamma)$-thick.
Also $i$ is $\beta$-perfect if it is not $\beta$-faulty.
\end{definition}

Now, we can prove a simple corollary from Lemma \ref{wdfsdfsdfsdf}
\begin{lemma}\label{ewlkfmwlekfmlwekf}
The number of $(\lambda, \phi, \tau, \upsilon)$-thick rows is at most:
$${G_k \upsilon^2 \over \lambda^{k-1}t_{\alpha} \phi {\tau}}$$
\end{lemma}
\ifNOTtalk
\begin{proof}
For every $(\lambda,\phi, \tau,\upsilon)$-faulty row $i$
there must be at least one $(\lambda, \phi\upsilon^{-1}, \tau)$-dense row in the range $\{i,\dots, i+\upsilon\}$.
Indeed if none of these rows are dense then for each $a \in [i, i+1, i+2... i + \upsilon]$:
$$
| \{  l: f_{l}(\alpha,a) > {\tau}, l \in T_{\lambda} \}| \le t_{\alpha} \phi\upsilon^{-1}.
$$
By summing up, we conclude that $i$ cannot be faulty.
Thus, the number of $(\lambda, \upsilon, \phi, \alpha, \tau)$-faulty rows is at most $\upsilon$ times the number of
$(\lambda, \phi\upsilon^{-1}, \alpha, \tau)$-dense rows. The last number is bounded by Lemma \ref{wdfsdfsdfsdf} as
$${G_k \upsilon \over \lambda^{k-1}t_{\alpha} \phi {\tau}}$$
Thus, the lemma follows.
\end{proof}
\fi

Let $\gamma\ge \beta$ be a parameter.
\begin{corollary}\label{ewrrkfkjlkjwefelkjkweflkjkwef}
The number of $(TR, \xi, \beta-7, 2^{\gamma})$-thick rows is at most
$f_1 \over 2^{\beta+\mu\gamma}$ for $k \ge 7$.
\end{corollary}
\ifNOTtalk
\begin{proof}
Repeat the proof of Corollary \ref{ewrferfwefwefwerfwerfwerfewf} using Lemma \ref{ewlkfmwlekfmlwekf} instead of Lemma \ref{wdfsdfsdfsdf}
and using ``faulty row'' instead of ``dense row.'' This introduces an additional factor of $4^{\gamma}$. The effect of the
change is neutralized by increasing the lower bound on $k$ by $2$: from $k\ge 5$ to $k\ge 7$.
\end{proof}
\fi
%


As a result, we can repeat the proof of Lemma \ref{erregregegergregregrft} by replacing ``great'' with ``perfect,''
``bad'' with ``faulty,'' and ``dense'' with ``thick.''

\begin{lemma}\label{erregregegergregregrft1}
Let $k\ge 7$. There exists a pair $\alpha,\beta$ such that the number of perfect rows that are from $S_{\beta}(\alpha)$ is at least
$$
{f_1 \over 2^{\beta+0.9\mu\beta}}.
$$
\end{lemma}

Our goal is to show that the arguments from Section \ref{dsfknskdjfkjsdfjksdfjk} will still be correct.
Specifically, we will repeat the proof of Lemma \ref{sdfljsdldjfljeflkjfd}.
First, we will discuss the first and second phases of the Signature creation process.
\begin{lemma}\label{sdfkjwesfkjwefkj}
Consider the event that $1$ has been chosen by the $z$-th player of $i$-th team.
That is, the event $C_{i,z}$ from the proof of Lemma \ref{sdfljsdldjfljeflkjfd} is true.
Then the probability of a non heavy element colliding with $1$ during the first and the second phase is bounded by $0.01$.
\end{lemma}
\ifNOTtalk
\begin{proof}
W.l.o.g., assume that $1$ has been sampled at row $i$ by player $z$. Let $l\neq 1$ be another element such that $f_{l}(\alpha,j) > 0$ for any $j$ such that $i
\le j \le i+2^{\gamma}$. The probability that $l$ will collide with $1$ is bounded by Fact \ref{sdgfdfgdfgdsg}. During the $\gamma$-th round, there are at most $2^{\gamma}t_{\alpha}$ such $l$. By using union bound, the conditional probability that any one of these elements will collide with $1$ is at most
$$
{2^{\beta-\Psi}\over t_{\alpha}}{1\over 2^{\varrho\gamma}}2^{\gamma}t_{\alpha} \le 2^{\beta-(\varrho-1)\gamma}.
$$
For $\varrho > 100$ and $\beta\le \gamma$ we have that the probability of collision for any $\gamma$ is at most $0.01$.
\end{proof}
\fi

In addition to the elements with large frequency, it is possible that sufficiently many elements with low frequency collide such that the
total frequency will be large. We bound the probability of that event.
\begin{lemma}
Consider the case when the event $C_{i,z}$ from the proof of Lemma \ref{sdfljsdldjfljeflkjfd} is true.
Let $Y_{i,z',\gamma}$ be the sum of all frequencies of all elements $l\notin T_{TR}$ that appear in rows $i,\dots, i+2^\gamma$ and that
agree on the signature and on the hash function with the sample of the $z'$-th player in the $i$-th row.
Let $Q_{i,z}$ be the event that the total sum of all $Y_{i,z',\gamma}$ of players that can play with $z$-th player
is larger than $2^\gamma$ for any $\gamma \le \log\log n$.
Then $P(Q_{i,z}\mid C_{i,z}) \le 0.01.$
\end{lemma}
\ifNOTtalk
\begin{proof}

Fix $\gamma$.
$$
E(Y_{i,z',\gamma}) \le 2^{\gamma} t_{\alpha} 2^{\gamma} {2^{\beta-\Psi} \over t_{\alpha}}{1\over 2^{\rho \gamma}} \le {1\over 2^{(\rho-3)\gamma}}.
$$
Summing over $\gamma$ we obtain the result.


\end{proof}
\fi

\noindent Now, we will discuss the third phase of Signature creation.

\begin{lemma}
Consider the case when the event $C_{i,z}$ from the proof of Lemma \ref{sdfljsdldjfljeflkjfd} is true.
Let $L_{i,z}$ be the event that $1$ will not beat one of its teammates after the third phase.
Then $P(L_{i,z}\mid C_{i,z}) \le 0.01.$
\end{lemma}
\ifNOTtalk
\begin{proof}
An estimated frequency of any element in the third phase is upper-bounded by its real frequency.
This follows from the definition of the third phase.
At the same time the frequency that can be lost for $1$ is bounded by $(\log\log n)  2^{4(\log\log n)}$.
Here the first number bounds the number of rounds and the second number bounds the maximum frequency to lose per round.
At the same time the frequency that will be collected in the second phase is at least $2^{10\log\log n}$.
Thus, we can repeat the analysis of Section \ref{dfvsdrgdrgrgrg} for the rounds in the third phase with only one change.
The counter will be $2^{\gamma - \alpha +\eta - 2}$ instead of $2^{\gamma - \alpha +\eta-1}$.
This change does not affect the correctness. Also, the analysis for faulty rows will be the same as in Corollary \ref{ewrrkfkjlkjwefelkjkweflkjkwef}.

\end{proof}
\fi


Thus, the behavior of the algorithm for this set will be identical to the behavior without modification and thus the player that samples the heavy hitter, $1$, will be the winner.
In addition to storing the IDs of the samples in the previous rows we also need to store the IDs during execution of reservoir sampling.
Once sampled we will assign a signature instead of an ID.
For the same reasons as in Section \ref{dfsdfhssfhstrt} the change does not affect the correctness.

\newpage
\section{Removing the assumption that $F_1 \le C_2n$.}\label{asdfadsfasfasf}
In this section we will remove the assumption that $F_1 \le C_2n$.
Instead we will assume that we allow two passes over the stream.
In the first pass we will compute $F_1$ and approximate $F_0$. In the second pass we will subsample the stream using $p={n/F_1}$.
The expected length of the sample stream is at most $10n$ w.p. $0.9$.

The following lemma shows that the heavy element in the original stream remains a heavy element in the sampled stream.
The frequency of the found heavy element is $(1\pm \epsilon)pf_i$, with high probability, by Chernoff bound.

The following is a section from \cite{subsampling}.
\begin{theorem}\label{thr:1}
Let $D$ be a stream and $i$ be a heavy element w.r.t. $F_k$ on $D$. Let $k\ge 1$ and let $p\ge \mu^{-1} = {F_0/F_1}$. Then there exists a constant $c_k$ such that with a constant probability,
$i$ is a $c_k$-heavy element w.r.t. $F_k$ on $D_p$.
\end{theorem}
\ifNOTtalk
\begin{proof}
By Chernoff bound, the frequency of $i$ in  $D_p$ is at least $(1-\epsilon)pf_i$ with high probability.
By Fact \ref{fact:5}, the $k$-th frequency moment of $D_p$ is bounded by $\alpha_k \mu^{-k}\sum_{i=1}^n v_i^k$. Thus, $i$ is a heavy element.
\end{proof}
\fi

\begin{fact}\label{fact:5}
Let $V \in {\left(Z^{+}\right)}^n$ be a vector with strictly positive integer entries $v_i$.
Let $\mu = {1\over n}\sum_{i=1}^n v_i$. Note that $\mu \ge 1$.
Let $X_i \sim B(v_i, \mu^{-1})$ and $X = \sum_{i=1}^n X_i^k.$
Then there exists a constant $\alpha_k$ that depends only on $k$ such that
$$
P(X > \alpha_k \mu^{-k}\sum_{i=1}^n v_i^k) < 0.1.
$$
\end{fact}
\ifNOTtalk
\begin{proof}
By Lemma \ref{fact:2}
$$
E(X_i^k) \le \beta_k ((\mu^{-1}v_i)^k + 1)
$$
Thus,
$$
E(X) < \beta_k (\mu^{-k}\sum_{i=1}^n v_i^k) + \beta_kn.
$$
Also, by the H\"{o}lder inequality
$$
{\sum_{i=1}^n v_i \over n^{1-1/k}} \le (\sum_{i=1}^n v_i^k)^{1/k}
$$
Thus,
$$
n^{1/k} = \mu^{-1}{\sum_{i=1}^n v_i \over n^{1-1/k}} \le \mu^{-1}(\sum_{i=1}^n v_i^k)^{1/k}
$$
Finally, $n < (\mu^{-k}\sum_{i=1}^n v_i^k)$.
We conclude the proof by putting $\alpha_k = 200\beta_k$ and applying Markov's inequality.
\end{proof}
\fi

\begin{lemma}\label{fact:2}
Let $X\sim B(N,p)$. There exists a constant $\beta_k$ that depends only on $k$ such that if $Np \ge 1$ then
\begin{equation}\label{eq:6}
E(X^k) \le \beta_k(Np)^k,
\end{equation}
and if $Np < 1$ then
\begin{equation}\label{eq:eeew}
E(X^k) \le \beta_k.
\end{equation}
\end{lemma}
\ifNOTtalk
\begin{proof}
Let $S(k,l)$ be a Stirling number of the second kind and let $B_k$ be the $k$-th Bell number (see \cite{Knuth:1997:ACP:260999} for the definition).
Using $(3.5)$ and $(1.247)$ from \cite{johnson2005univariate}, we can write:
\begin{equation}\label{eq:ffdgfdgfdg}
E(X^k) = \sum_{l=0}^k S(k,l){N!p^l\over (N-l)!}.
\end{equation}
Recall that $B_k = \sum_{l=0}^k S(k,l)$. Thus,
\begin{equation}\label{eq:rgtretertre}
E(X^k) \le B_k \sum_{l=0}^k (Np)^l.
\end{equation}
If $Np \ge 1$ then
\begin{equation}\label{eq:erreerreer}
E(X^k) \le (k+1)B_k(Np)^k,
\end{equation}
and if $Np < 1$ then
\begin{equation}\label{eq:errefdfgsd}
E(X^k) \le (k+1)B_k.
\end{equation}
We conclude\footnote{The recent bound on Bell number is
$
B_k <\left({
0.792k\over
ln(k + 1)}
\right)^k
$  due to \cite{Berend}} our proof by defining $\beta_k = (k+1)B_k$.
\end{proof}
\fi

\newpage
\section{Winning Pairs}\label{kljdflkjdsfsfkldfknkdffknsdkksd}
The following is a section from \cite{DBLP:journals/corr/abs-1212-0202}, provided for completeness.

\begin{definition}\label{def: dkjfjrglrlttrg}
Let $U = \{u_1,\dots,u_t\}$ and $W=\{w_1,\dots,w_t\}$ be two
sequences of non-negative integers. Let $(i,j)$ be a pair such that $1\le i\le
t$ and $1\le j\le u_i$. Denote $(i,j)$ as a \emph{losing} pair (w.r.t. sequences $U,W$) if there exists $h, i\le h\le t$ such that:
$$
-j + \sum_{s=i}^h(u_s-w_s) < 0.
$$
Denote any pair that is not a losing pair as a a \emph{winning} pair.
\end{definition}

In this section we consider the following pair $(U,W)$ of sequences. For $i=1,\dots, r$
let $u_i = f_{1,i}$ and $w_i = \lambda$.

\begin{definition}\label{def: dkjfjrbfghfghgglrlttrg}
Let $U = \{u_1,\dots,u_t\}$ and $W=\{w_1,\dots,w_t\}$ be two
sequences of non-negative integers. Let $1\le h < t.$ Let $U',W'$ be two sequences of size $t-h$
defined by $p'_i = u_{i+h}$, $q'_i = w_{i+h}$ for $i=1,\dots, t-h$.
Denote $U',W'$ as the $h$-tail of the sequences $U,W$.
\end{definition}

\noindent
\begin{fact}\label{fct:rtheyjtyjukyukyuk}
If $(i,j)$ is a winning pair w.r.t. the $h$-tail of $U,W$ then
$(i+h,j)$ is a winning pair w.r.t. $U,W$.
If $(i,j)$ is a winning pair w.r.t. the $h$-tail of $U,W$ then
$(i,j)$ is a winning pair w.r.t. $U,W$.
\end{fact}

\ifNOTtalk

\begin{proof}
Follows directly from Definitions \ref{def: dkjfjrglrlttrg} and \ref{def: dkjfjrbfghfghgglrlttrg}.
\end{proof}

\fi

\begin{lemma}\label{fct:erergerger}
If $\sum_{s=1}^t(u_s-w_s) > 0$ then there exist at least $\sum_{s=1}^t(u_s-w_s)$ winning pairs.
\end{lemma}

\ifNOTtalk

\begin{proof}

We use induction on $t$. For $t=1$, any pair $(1,j)$
is winning if $1\le j\le u_1-w_1$. Consider $t>1$ and apply the following case analysis.

\begin{enumerate}
\item Assume  that there exist $1\le h< t$ such that $\sum_{s=1}^h(u_s-w_s) \le 0$. Consider the $h$-tail of $U,W$. By
induction and by Fact \ref{fct:rtheyjtyjukyukyuk}, there exist at least
$\sum_{s=h+1}^t(u_s-w_s) \ge \sum_{s=1}^t(u_s-w_s) $ winning pairs w.r.t. $U,W$.

\item Assume  that $(1,u_1)$ is a winning pair; it follows that $(1,j),\ j<u_1$ is a winning pair as well. If $\sum_{s=2}^t(u_s-w_s) > 0$ then, by
induction and by Fact \ref{fct:rtheyjtyjukyukyuk}, there exist at least
$\sum_{s=2}^t(u_s-w_s)$ winning pairs of the form $(i,j)$ where $i>1$. In total there are $u_1 + \sum_{s=2}^t(u_s-w_s) \ge \sum_{s=1}^t(u_s-w_s)$ winning pairs w.r.t. $U,W$.  The case when $\sum_{s=2}^t(u_s-w_s) < 0$ is trivial.

\item Assume that $(1), (2)$ do not hold. Then $u_1>0$. Indeed otherwise $u_1-w_1\le 0$ and thus $(1)$ is true. Also  $(1,1)$ is a winning pair. Indeed, otherwise there exists $1\le h < t$ such that $-1 + \sum_{i=1}^h(u_i-w_i)<0$. All numbers are integers thus $\sum_{i=1}^h(u_i-w_i)\le 0$ and $(1)$ is true.
Thus, $(1,1)$ is a winning pair and $(1,u_1)$ is not a winning pair (by $(2)$). Therefore there exist $1<u\le u_1$ such that $(1,u-1)$ is a winning pair and $(1,u)$ is not a winning pair. In particular, there exists $1\le h<t$ such that
$$
-u + \sum_{s=1}^h(u_s-w_s) <0.
$$
On the other hand $(1,u-1)$ is a winning pair thus
$$
0\le 1-u + \sum_{s=1}^h(u_s-w_s).
$$
All numbers are integers and thus we conclude that
$$
\sum_{s=1}^h(u_s-w_s) = u-1.
$$
Consider the $h$-tail of $U,W$.
By induction, there exists at least
$$
\sum_{i=h+1}^{t}(u_i-w_i) = \sum_{i=1}^{t}(u_i-w_i) - (u-1)
$$
winning pairs w.r.t. the $h$-tail of $U,W$. By Fact \ref{fct:rtheyjtyjukyukyuk} there exist at least as many winning pairs w.r.t. $U,W$ of the form $(i,j)$ where $i>1$. By properties of $u$ there exists an additional $(u-1)$ winning pairs of the form $(1,j), j\le u-1$.  Summing up we obtain the fact.
\end{enumerate}
\end{proof}

\fi
\end{document}